\documentclass[11pt]{article}

\usepackage{setspace} 
\usepackage{titlesec}
\titlespacing*{\paragraph}{0pt}{0.05\baselineskip}{0.5em}

\usepackage[letterpaper,
            top=1in,
            bottom=1in,
            left=1in,
            right=1in]{geometry}

\usepackage{cancel}
\usepackage{paralist}
\usepackage[shortlabels]{enumitem}
\setlist{nosep}
\setlist[itemize]{topsep=0.25\baselineskip}
\setlist[enumerate]{topsep=0.25\baselineskip}
\usepackage{graphicx}
\usepackage{subcaption}
\usepackage{xcolor}
\usepackage{bm}
\usepackage{mathtools}

\usepackage[many]{tcolorbox}

\usepackage{tikz}
\usetikzlibrary{3d,calc,quotes,angles,patterns,tikzmark,shapes.misc}

\usepackage{algorithm}
\usepackage{algpseudocode}

\usepackage{tocloft}
\renewcommand{\cftsecleader}{\cftdotfill{\cftdotsep}}
\usepackage{amsmath, amsthm, amssymb}
\usepackage{thmtools}

\usepackage[T1]{fontenc}
\usepackage{cabin} 
\usepackage[theoremfont, tighter]{newpxtext}
\usepackage[varbb, varg]{newpxmath}
\usepackage[cal=euler, bb=dsserif]{mathalpha}

\usepackage{microtype}
\UseMicrotypeSet[protrusion]{basicmath}

\definecolor{defcol}{RGB}{10, 105, 75}
\newcommand{\defemph}[1]{\textcolor{defcol}{\emph{#1}}}

\renewcommand{\leq}{\leqslant}
\renewcommand{\geq}{\geqslant}
\renewcommand{\le}{\leqslant}
\renewcommand{\ge}{\geqslant}

\usepackage{hyperref}
\definecolor{refcol}{RGB}{200, 80, 60}
\definecolor{citecol}{RGB}{100, 80, 180}
\hypersetup{
    colorlinks,
    linkcolor=refcol,
    citecolor=citecol,
    urlcolor={blue!80!black}
}
\usepackage[capitalize,nameinlink]{cleveref} 

\declaretheorem[style=plain, numberwithin=section, name=Theorem]{theorem}
\declaretheorem[style=plain, sibling=theorem, name=Lemma]{lemma}
\declaretheorem[style=plain, sibling=theorem, name=Corollary]{corollary}
\declaretheorem[style=plain, sibling=theorem, name=Proposition]{proposition}
\declaretheorem[style=plain, sibling=theorem, name=Claim]{claim}
\declaretheorem[style=plain, sibling=theorem, name=Fact]{fact}

\declaretheorem[style=definition, sibling=theorem, name=Definition]{definition}

\declaretheorem[style=remark, sibling=theorem, name=Remark]{remark}

\declaretheorem[style=plain, name=Theorem, Refname={Theorem,Theorems}]{stheorem}
\declaretheorem[style=plain, name=Lemma]{slemma}
\declaretheorem[style=definition, name=Definition, Refname={Definition,Definitions}]{sdefinition}

\makeatletter
\def\mkcal#1{%
  \expandafter\providecommand\csname c#1\endcsname{%
    \ensuremath{{\bm{\mathcal{#1}}}}%
  }%
}
\@tfor\next:=A,B,C,D,E,F,G,H,I,J,K,L,M,N,O,P,Q,R,S,T,U,V,W,X,Y,Z\do{%
  \expandafter\mkcal\next
}
\makeatother

\newcommand{\N}{\mathbb{N}} 
\newcommand{\Z}{\mathbb{Z}} 
\newcommand{\F}{\mathbb{F}} 

\newcommand{\Bits}{\{0, 1\}} 

\renewcommand{\hat}{\widehat}
\newcommand{\eps}{\varepsilon}

\DeclareMathOperator*{\E}{{\bm{\mathrm{E}}}}
\let\Pr\relax
\DeclareMathOperator*{\Pr}{{\bm{\mathrm{Pr}}}}
\DeclareMathOperator*{\Var}{{\bm{\mathrm{Var}}}}
\DeclareMathOperator*{\Cov}{{\bm{\mathrm{Cov}}}}

\newcommand{\Ber}{\mathrm{Ber}} 

\newcommand{\defeq}{\coloneq}
\newcommand{\eqdef}{\eqcolon}

\newcommand{\rest}[1]{{{\mkern-3mu\upharpoonright\mkern-2mu}_{ #1 }}} 
\newcommand{\wt}{\operatorname{wt}} 
\newcommand{\supp}{\operatorname{supp}} 
\newcommand{\rank}{\operatorname{rank}}
\newcommand{\Span}{\operatorname{span}}
\newcommand{\bias}{\operatorname{bias}}

\newcommand{\seq}{\subseteq}

\DeclarePairedDelimiter{\abs}{\lvert}{\rvert} 
\DeclarePairedDelimiter{\card}{\lvert}{\rvert} 
\DeclarePairedDelimiter{\norm}{\lVert}{\rVert} 

\newcommand{\indicator}[1]{\mathbb{1}_{\langle #1 \rangle }}

\newcommand{\ox}{\otimes}

\newcommand{\half}{\frac{1}{2}}

\newcommand{\lin}{\mathscr{L}}

\newcommand{\KL}{\text{\ttfamily KL}}
\newcommand{\tv}{\text{\ttfamily TV}}

\newcommand{\ACzero}{\mathbf{AC}^{0}}
\newcommand{\NCzero}{\mathbf{NC}^{0}}

\newcommand{\term}{\bold{K}}
\allowdisplaybreaks

\title{On Sampling Lower Bounds for Polynomials}

\author {
  Mohammad Mahdi Khodabandeh\\
  \texttt{mmk25@sfu.ca}\\
  Simon Fraser University
  \and
  Igor Shinkar\\
  \texttt{ishinkar@sfu.ca}\\
  Simon Fraser University
}
\date {\today}

\begin{document}

\maketitle

\begin{abstract}
    In this work, we continue the line of research on the complexity of distributions~\cite{Viob}, and study samplers defined by low degree polynomials. An $n$-tuple $\mathcal{P} = (P_1,\dots, P_n)$ of functions $P_i \colon \mathbb{F}_2^m \to \mathbb{F}_2$ defines a distribution over $\{0, 1\}^n$ in the natural way: draw $X$ uniformly at random from $\mathbb{F}_2^m$ and output $(P_1(X),\dots, P_n(X)) \in \{0,1\}^n$.
    
    We show that when $\mathcal{P}$ is defined by polynomials of degree $d$, the total variation distance of $\mathcal{P}$ from the product distribution $\mathrm{Ber}(1/3)^{\otimes n}$ is $1-o_n(1)$, where $o_n(1)$ is a vanishing function of $n$ for any constant degree $d$. For small values of $d$, we show the following concrete bounds.
    \begin{itemize}
        \item For $d=1$ we have $\|\mathcal{P} - \mathrm{Ber}(1/3)^{\otimes n}\|_{\mathtt{TV}} \geq 1 -  \exp(-\Omega(n))$.
        \item For $d=2$ we have $\|\mathcal{P} - \mathrm{Ber}(1/3)^{\otimes n}\|_{\mathtt{TV}} \geq 1 -  \exp(-\Omega(\log(n)/\log\log(n)))$. 
        \item For $d=3$ we have $\|\mathcal{P} - \mathrm{Ber}(1/3)^{\otimes n}\|_{\mathtt{TV}} \geq 1 -  \exp(-\Omega(\sqrt{\log\log(n)}))$. 
    \end{itemize}
    Our results extend the recent lower bound results for sampling distributions, which have mostly focused on local samplers, small depth decision trees, and small depth circuits~\cite{lovett2011bounded, beck2012large, viola2023new, filmus2023sampling, kane2024locality}.

    As part of our proof, we establish the following result, that may be of independent interest: for any degree-$d$ polynomial $P \colon \mathbb{F}_2^m \to \mathbb{F}_2$ it holds that $\Pr_X[P(X) = 1]$ is bounded away from $1/3$ by some absolute constant $\delta = \delta_d>0$. Although the statement may seem obvious, we are not aware of an elementary proof of this.

    The proof techniques rely on the structural results for low degree polynomials~\cite{green2007distribution,kaufman2008worst, haramaty2010structure}, saying that any biased polynomial of degree $d$ can be written as a function of a small number of polynomials of degree $d-1$.
\end{abstract}

\newpage

\setlength{\cftbeforesecskip}{0pt}
\renewcommand\cftsecfont{\mdseries}
\renewcommand{\cftsecpagefont}{\normalfont}
\renewcommand{\cftsecleader}{\cftdotfill{\cftdotsep}}
\setcounter{tocdepth}{1}
\tableofcontents

\newpage
\section{Introduction}\label{sec:intro}

\subsection{Background and previous work}

The seminal work of Viola \cite{Viob} initiated the systematic study of a topic referred to as \emph{the complexity of distributions}. In this framework, one has a target distribution $\cT$ over $\Bits^n$ that one wishes to sample from using uniformly random bits. Sampling is then the task of designing an algorithm $A$ that, given a uniformly random input bit string $X\sim\Bits^m$, outputs a string $A(X)\in\Bits^n$ such that the resulting output distribution is close to $\cT$ in statistical distance. The class of algorithms $A$ under consideration, i.e., the model of computation, and the target distribution determine the nature of the problem in this context. One is then led to both upper- and lower-bound questions: to either show the existence of an algorithm that (approximately) samples the distribution $\cT$, or to show that every such algorithm fails to generate a distribution close to $\cT$.

Over the past fifteen years this area has received considerable attention \cite{Viob, lovett2011bounded, beck2012large, de2012extractors, viola2016quadratic, viola2020sampling, goos2020lower, chattopadhyay2022space, viola2023new, filmus2023sampling, kane2024locality}. Viola \cite{Viob} started by showing the first impossibility result that local functions $f\colon \Bits^m\to\Bits^n$, where every output bit $f(X)_i$ depends on only a few input bits, are unable to sample a uniformly random string from a Hamming slice $H_k=\{x\in\Bits^n : \sum_ix_i = k\}$. The work also featured an application of such sampling lower bounds to lower bounds for succinct data structures which were extended in \cite{kane2024locality}. Subsequent works \cite{watts2023unconditional, kane2024locality} explored applications of a different kind, namely separations between the sampling power of quantum and classical circuits. The hardness of sampling was also used to construct explicit codes \cite{shaltiel2024explicit}.

The following typical example illuminates why sampling is easier than computing in general, and therefore, proving lower bounds is a more difficult task. For uniformly random $(X_1,X_2,\dots, X_n) \in \Bits^n$ the distribution $(X_1,X_2,\dots, X_n, X_1\oplus \dots \oplus X_n)$ can be alternatively expressed as $(X_1\oplus X_2, X_2\oplus X_3, \dots, X_{n-1}\oplus X_{n}, X_{n}\oplus X_1)$ where $\oplus$ denotes sum mod $2$.
Note that although the former view famously has no sub-exponential size circuit in $\ACzero$ (\cite{furst1984parity, ajtai198311, hastad1986almost}), the latter form can be implemented with a $2$-local circuit.

Examples of this kind brought attention to the study of the sampling power of local maps or related models such as shallow decision forests. On this front, Filmus, Leigh, Riazanov, and Sokolov \cite{filmus2023sampling} show a lower bound against shallow decision forests for sampling from a Hamming slice $H_{o(n)}$. They also posed a conjecture characterizing the power of $\NCzero$ in sampling uniform distributions over a symmetric support, which was resolved by Kane, Ostuni, and Wu \cite{kane2025locally} in the affirmative. They later extended this result to characterize all symmetric distributions sampled by local maps \cite{kane2025symmetric}. Horacsek, Lee, Shinkar, Viola, and Zhou \cite{horacsek2025sourcedecoding} in a recent result showed that $\NCzero$ can sample a product of dyadic Bernoullies using nearly the information-theoretically minimum number of uniform random bits.

\subsection{This work}

We turn our focus to proving lower bounds for the class of degree $d$ polynomials over $\F_2$, a class that contains within itself the class of $d$-local maps and depth-$d$ decision forests.
A \emph{degree $d$ distribution} over $\Bits^n$ is defined by a parameter $m \in \N$ and $n$ polynomials $P_1,\dots, P_n \colon  \F_2^m\to \F_2$ of degree $d$. We sample from such a distribution $\cP = (P_1,\dots, P_n)$ by sampling a uniformly random input $X\sim \F_2^m$ and outputting $(P_1(X), \dots, P_n(X))$.
The main result of this paper shows that for any constant $d$ any degree-$d$ distribution is $1-o_n(1)$-far from $\Ber(1/3)^{\ox n}$ in total variation distance, where $o_n(1)$ is a vanishing function of $n$ for any constant $d$. For small degrees $d=1,2,3$ we show concrete bounds on the vanishing term $o_n(1)$.%
\footnote{The parameter $1/3$ is fixed throughout the paper for simplicity, and can be replaced by any non-dyadic $\rho \in (0,1)$.}

We also study the class of distributions generated by \emph{bounded rank} polynomials where each $P_i$ is a function of bounded $\rank_d$. Rank, as used here, comes from the work of Green and Tao \cite{green2007distribution} who defined it to study its relationship with pseudorandom properties of polynomials and Gowers norms. A function $P \colon \F_2^m\to\F_2$ is said to have $\rank_d(P) \le r$ if there exist $r$ polynomials $Q_1(x),\dots,Q_r(x)$ of degree at most $d$ and some $\Gamma\colon\F_2^r\to\F_2$ such that $P(x)=\Gamma(Q_1(x),\dots, Q_r(x))$ for all $x$. Distributions defined by functions of small $\rank_d$ essentially interpolate between degree $d$ and $d+1$ distributions. Indeed, our strategy for proving lower bounds against low degree polynomials will be to consider distributions defined by bounded rank functions. Specifically, we will first consider distributions $(P_1,\dots, P_n)$ of bounded linear rank, i.e., where each $P_i$ satisfies $\rank_1(P_i)\le r$ for which we show a lower bound of $1-\exp(-2^{-O(r^2)}n)$ on the distance from $\Ber(1/3)^{\ox n}$. Note that these distributions already generalize $r$-local distributions. Then, we will consider functions of bounded quadratic rank for which we prove a $1 - \exp(-\log(n)^{\Omega(1/r^2)})$ lower bound. These will help us prove a $1 - \exp(-\Omega(\sqrt{\log \log (n)}))$ lower bound against degree three distributions.

Analogous results have been shown for $d$-local maps in \cite{kane2024locality, viola2023new}, where the lower bound is against sampling Hamming slices. \cite{viola2023new} states the result against depth-$d$ decision forests for sampling from the Hamming slice $H_{n/h}$ for $h$ not a power of $2$, while the result of \cite{kane2024locality} extends it to any $h$. Many of our results can be translated from the setting of product distribution $\Ber(1/3)^{\ox n}$ to a sampling lower bound of $H_{n/3}$ against degree $d$ distributions with essentially no change.

While it is natural to expect that for any $\rho \in (0,1)$ similar lower bounds hold for $H_{\rho n}$ against degree $d$ distributions, in the product distribution setting $\Ber(\rho)^{\ox n}$ it is essential that the parameter $\rho$ is some non-dyadic, since for dyadic $\rho=a/2^d$ the distribution $\Ber(\rho)^{\ox n}$ can be perfectly sampled by polynomials of degree $d$.

An essential (and non-trivial) step in our results is to show that a single polynomial $P\colon \F_2^m\to\F_2$ cannot sample $\Ber(1/3)$ better than some constant $\Omega_d(1)$ that does not depend on $m$. Put differently, $\abs{\Pr_X[P(X)=1]-1/3}$ is bounded away from zero by some constant that only depends on $d$. In contrast, the analogous question for $d$-local distributions is trivial, since the probability that a $d$-local function outputs $1$ is $a/2^d$ for some integer $a$ which leaves a $\Omega(2^{-d})$ gap from $1/3$; however, a degree $d$ polynomial $P\colon\F_2^m\to\F_2$ can genuinely depend on all $m$ bits, and even for degree $d=2$ we can construct a quadratic polynomial with $\Pr[P(x) = 1] = 3/8$, which clearly is not of the form $a/2^d$. To show this, we use structural results for low degree polynomials from the works of Green and Tao \cite{green2007distribution}, Kaufman and Lovett \cite{kaufman2008worst}, and Haramaty and Shpilka \cite{haramaty2010structure}.
A related problem to this result is the weight distribution of binary Reed-Muller codes \cite{kasami1970weight, kasami1976weight, kaufman2012weight} which asks for given parameters $m,d\in \N, \eta\in(0,1)$ how many polynomials $P\colon \F_2^m\to\F_2$ of degree $d$ exist that satisfy $\Pr_X[P(X)=1]\le \eta$. Our result says that for a constant $\delta=\delta_d > 0$, the number of such polynomials does not change when $\eta$ is in the $\delta$-neighborhood of $1/3$, that is, $\eta\in [1/3-\delta, 1/3 + \delta]$.

As mentioned above, the parameter $1/3$ is fixed throughout the paper for simplicity only, and analogous distance lower bound results can be shown for the distribution $\Ber(\rho)^{\ox n}$ for any non-dyadic parameter $\rho \in (0,1)$. In fact, since for a single degree $d$ polynomial $\Pr[P(x) =1 ]$ is bounded away from, say $\rho = 1/2^{d+1}$, our distance lower bounds also extend to $\Ber(1/2^{d+1})^{\ox n}$.

At a high level, our results form a hierarchy where we show lower bounds against degree $1$, bounded $\rank_1$, degree $2$, bounded $\rank_2$, and so on, at the cost of progressively weaker distance guarantees. We conjecture that the correct lower bound for sampling $\Ber(1/3)^{\ox n}$ with distributions of bounded $\rank_d \le r$ should be $1-\exp(-\Omega_{d, r}(n))$.

One of the main technical tools in establishing our results is a \emph{Chebyshev lemma} stating the following. Suppose we are given a distribution $\cP = (P_1,\dots, P_n)$, where each $P_i$ can be written as $P_i=\Gamma_i(Q_{i1},\dots, Q_{ir})$ for some degree $d$ polynomials $Q_{i1},\dots,Q_{ir}$, that are \emph{almost pairwise independent}. That is, for each $i \neq j$ the $2r$-tuple $(Q_{i1},\dots, Q_{ir}; Q_{j1},\dots, Q_{jr})$ is a collection of \emph{almost independent} polynomials. Then, the distance of $\cP$ from the distribution $\Ber(1/3)^{\ox n}$ is $\norm{\cP - \Ber(1/3)^{\ox n}}_\tv \geq 1 - o_r(1) - o_n(1)$.

One issue is that, generally, polynomial distributions are not \emph{almost pairwise independent}. However, we can hope for a relaxed scenario of finding a \emph{subset} of the $P_i$'s so that the corresponding $(Q_{i1},\dots,Q_{ir})$ satisfy the required \emph{almost pairwise independence} for $P_i$'s in the subset. We will show a slight variant of that. Specifically, we show that any set of polynomials has a large subset of $P_i$'s such that their corresponding $(Q_{i1},\dots,Q_{ir})$ are almost pairwise independent \emph{up to a small collection of $Q^*_1,\dots, Q^*_c$ that appear in all of them}. In other words, the tuples $(Q_{i1},\dots,Q_{ir})$ form a \emph{sunflower} whose petals are almost pairwise independent.

In order to quantify \emph{almost independence} for a collection of polynomials, we use the notion of \emph{regularity}, which has been studied in the context of structure versus randomness of polynomials \cite{green2007distribution,kaufman2008worst}. Using similar machinery, we show that any collection of low degree polynomials contains a large subset that form a sunflower, whose petals are pairwise regular polynomials in the sense described above. This is of independent interest to us, and, to the best of our knowledge, has not appeared previously in the literature.
While sunflowers have appeared before in related contexts (see, e.g., \cite{filmus2023sampling,yu2024sampling}), our work appears to be the first to study such a sunflower structure for polynomials.

\subsection{Organization of the paper}

In \cref{sec:our-results} we give a broad overview of our results and techniques. \cref{sec:prelims} fixes our notation, introduces definitions, and presents several basic lemmas that will be used throughout the paper. In \cref{sec:chebyshev-lemma} we prove our \emph{Chebyshev lemma} and some of its variants, which will be used in the subsequent sections. In \cref{sec:pairwise-reg} we prove our key combinatorial lemma on \emph{sunflower pairwise regularization} that is used in proving our sampling lower bound for constant degree distributions. \cref{sec:lower-bound-for-a-single-poly} proves that a single polynomial cannot sample $\Ber(1/3)$ better than a constant distance; this section is extended to $\Ber(\rho)$ for any non-dyadic $\rho$ in \cref{sec:non-constructive-lower-bound-for-a-single-poly}. The remaining sections, \cref{sec:deg-1-lower-bound,sec:bounded-rank-1-lower-bound,sec:degree-2-lower-bound,sec:bound-rank2-lower-bound,sec:degree-3-lower-bound,sec:deg-d-lower-bound}, talk about lower bounds of sampling $\Ber(1/3)^{\ox n}$ against degree $1$, bounded $\rank_1$, degree $2$, bounded $\rank_2$, degree $3$, and general degree $d$ distributions.

\section{Our results}\label{sec:our-results}

In this section we present our results and describe the ideas and techniques utilized to prove them. In some cases we present an informal statement in favor of a cleaner exposition. We close this section with some open problems.

\subsection{A single polynomial cannot approximate $\Ber(1/3)$ well}

As discussed in the introduction, the following theorem lies at the heart of all of our lower bounds.
\begin{stheorem}[A gap for one polynomial; formal version in \cref{thm:exp-low-deg-far-from-third}]\label{main-thm:one-polys-gap}
    Let $P \colon \F_2^m \to \F_2$ be a polynomial of degree $d$.
    Then
    \[
    \abs{\Pr_{X \sim \F_2^m}[P(X) = 1] - \frac{1}{3}} \geq \delta,
    \]
    for some constant $\delta = \delta_d>0$ that depends only on $d$.
\end{stheorem}

Although the statement may appear almost self-evident at first glance, its proof is non-trivial even for degree as small as three. For quadratic polynomials, the proof can be obtained from the structural results of quadratic polynomials of Dickson (see, e.g., \cite[Theorem 6.30]{LidlNiederreiter1996}). Thus, quite naturally, our argument ultimately relies on regularizing techniques from \cite{green2007distribution,kaufman2008worst} and the corresponding structural results. For degree three, we give an alternative proof using the specific structural results for cubic polynomials of Haramaty and Shpilka~\cite{haramaty2010structure}.

\begin{remark}
We remark that throughout the paper, the choice of $1/3$ is only a matter of simplicity and convenience. In \cref{sec:non-constructive-lower-bound-for-a-single-poly}, we show that $1/3$ in \cref{main-thm:one-polys-gap} can be replaced with any non-dyadic\footnote{A number of the form $a/2^b$ where $a$ and $b$ are integers is called \defemph{dyadic}.} parameter $\rho$, and our sampling lower bounds for $\Ber(\rho)^{\ox n}$ hold for any $\rho \in (0,1)$ satisfying \cref{main-thm:one-polys-gap}.
\end{remark}

Before proceeding to other results, let us first formally define a polynomial distribution.

\begin{sdefinition}[Degree-$d$ distributions]
An $n$-tuple $\cP=(P_1,\dots,P_n)$ where each $P_i\colon\F_2^m\to\F_2$ is a polynomial of degree at most $d$, defines a distribution over $\Bits^n$ in the natural way: draw $X \sim \F_2^m$ uniformly at random, and output $(P_1(X),\dots,P_n(X)) \in \Bits^n$.
\end{sdefinition}

In our setting, the number of input variables $m$ may be arbitrarily large. Observe that the polynomials $P_1,\dots,P_n$ are fed the same input $X$ which allows the $n$ outputs to be correlated.

\subsection{Lower bounds against linear distributions}

We start with the following simple claim, showing a lower bound for degree-$1$ distributions. We include this result for completeness rather than for novelty.
\begin{stheorem}[Degree-$1$ distributions; formal version in \cref{sec:deg-1-lower-bound}]\label{main-thm:linear}
    Let $\cL = (L_1,\dots,L_n)$ be a degree-$1$ distribution.
    Then $\norm{\cL - \Ber(1/3)^{\ox n}}_\tv \geq 1 - \exp(-cn)$ for some absolute constant $c>0$.
\end{stheorem}

The (simple) proof of this theorem illustrates a strategy that is used frequently in proving lower bounds for sampling: reducing to a win-win scenario. The argument for this theorem goes as follows. Let $D$ be the maximum number of linearly independent $L_i$'s. The key is to analyze what happens when $D$ is large and small respectively. If $D$ is large, say, more than $n/2$, then the $D$ linearly independent coordinates form a uniform distribution over $\Bits^D$, thus giving $1-2^{-\Omega(n)}$ distance from the target distribution $\Ber(1/3)^{\ox n}$. On the other hand, if $D$ is small, then conditioning on the values of the $D$ linearly independent $L_i$'s fixes the value of all other coordinates. In other words, $\cL$ is expressed as a convex combination of $2^D$ many distributions, each with singleton support. This case also gives a $1-2^{-\Omega(n)}$ lower bound, using a certain union bound claim for convex combinations of distributions.

Our first interesting result is the following theorem, showing a sampling lower bound for bounded $\rank_1$ distributions. These are distributions $\cP = (P_1,\dots, P_n)$ where each $P_i\colon\F_2^m\to\F_2$ can be described by a bounded number of linear/degree-$1$ polynomials. Already at this point, we are considering a more general class of distributions than local maps. For $r$-local maps the works of \cite{viola2023new,kane2024locality} essentially prove a $1-\exp(-\Omega_r(n))$ lower bound for sampling $\Ber(1/3)^{\ox n}$. We show a similar bound for distributions whose $\rank_1$ is upper bounded by $r$.

\begin{stheorem}[Bounded $\rank_1$ distributions; formal version in \cref{thm:bounded-rank-1}]\label{main-thm:bounded-linear-rank}
    Let $\cP = (P_1,\dots,P_n)$ be a distribution, where each $P_i\colon\F_2^m\to\F_2$ can be written in the form $P_i=\Gamma_i(L_{i1},\dots, L_{ir})$, where $\Gamma_i \colon \Bits^r \to \Bits$ is some function and $L_{i1},\dots, L_{ir} \colon \F_2^m \to \F_2$ are degree-$1$ polynomials.
    Then $\norm{\cP - \Ber(1/3)^{\ox n}}_\tv \geq 1 - \exp(-2^{-Cr^2}n)$ for some absolute constant $C>0$.
\end{stheorem}

The proof of the theorem is an induction on $r$, and has a similar flavor to the lower bound for $r$-local functions in \cite{kane2024locality}, generalized to account for $r$ linear functions (rather than $r$ input bits) using linear algebraic tools. We define a bipartite graph, where on the right vertex set we have our $P_i$'s, and on the left vertex set we have all linear functions. We then connect each polynomial to the subspace of dimension $r$ that it depends on; in the notation of \cref{main-thm:bounded-linear-rank}, $P_i$ is adjacent to all $L\in \Span(L_{i1},\dots, L_{ir})$. Then the argument considers a maximal subspace $W$ (a subset of the left vertex set) whose dimension is not too large and which expands well relative to that dimension. If $W$ is adjacent to many polynomials, then by conditioning/fixing the values of the linear functions in $W$ we reduce the rank of its neighbors, which allows for applying induction. Otherwise, the neighborhood of $W$ is rather small, and its maximality ensures that we can find a large subset of polynomials that depend on linearly independent subspaces, which reduces to the scenario where a large subset of $P_i$'s are \emph{independent}, and hence the subset corresponds to a product distribution with each coordinate having $\Omega(2^{-r})$ distance from $\Ber(1/3)$.

\subsection{Lower bounds against quadratic distributions}

We would like to extend our lower bounds in \cref{main-thm:linear,main-thm:bounded-linear-rank} to quadratic distributions. Suppose we are given a distribution $\cQ=(Q_1,\dots,Q_n)$ defined by quadratic polynomials. Note that if, say, at least $n/2$ polynomials $Q_i$ have bounded $\rank_1(Q_i)$, then we can simply apply \cref{main-thm:bounded-linear-rank}. Hence, it remains to handle the case, where many $Q_i$'s have high $\rank_1$.

Let us assume for simplicity that all $Q_i$'s have high $\rank_1$.
Note first that for parameters $r$ and $s$ one of the following two cases must hold.
\begin{enumerate}[(a)]
    \item There is a subset $S \seq [n]$ of $s$ polynomials such that $\rank(Q_i+Q_j) \geq r$ for all $i \neq j \in S$.
    \item There exists some $i^*$ and a set $S \seq [n]$ of size $\geq t \geq n/s$ such that $\rank(Q_{i^*}+Q_j) \leq r$ for all $j \in S$.
\end{enumerate}
Indeed, this can be viewed as a graph-theoretic argument, where the vertices of the graph correspond to the $n$ polynomials, and $(i,j)$ is an edge if and only if $\rank(Q_i+Q_j) \leq r$. Then, the first case corresponds to the graph having an independent set of size $s$ and the second case corresponds to the graph having a vertex of degree at least $n/s$.

\begin{figure}[ht]
    \centering

    \begin{subfigure}[t]{0.45\textwidth}
        \centering
        \begin{tikzpicture}[
    vertex/.style={circle, draw, inner sep=1pt, minimum size=12pt}
]
    \node[vertex] (Q1) at (90:2.0)   {$Q_1$};
    \node[vertex] (Q2) at (18:2.0)   {$Q_2$};
    \node[vertex] (Q3) at (-54:2.0)  {$Q_3$};
    \node[vertex] (Qt) at (162:2.0)  {$Q_s$};

    \coordinate (D1) at (-118:2.0);
    \coordinate (D2) at (-126:2.0);
    \coordinate (D3) at (-134:2.0);

    \fill (D1) circle (1.2pt);
    \fill (D2) circle (1.2pt);
    \fill (D3) circle (1.2pt);

    \draw[dotted] (Q1) -- (Q2);
    \draw[dotted] (Q1) -- (Q3);
    \draw[dotted] (Q1) -- (Qt);
    \draw[dotted] (Q2) -- (Q3);
    \draw[dotted] (Q2) -- (Qt);
    \draw[dotted] (Q3) -- (Qt);

    \foreach \V in {Q1,Q2,Q3,Qt} {
        \draw[dotted, shorten >=7pt] (\V) -- (D1);
        \draw[dotted, shorten >=7pt] (\V) -- (D2);
        \draw[dotted, shorten >=7pt] (\V) -- (D3);
    }
\end{tikzpicture}
        \caption{\textbf{Large independent set.} The vertices form an independent set. The tuple $(Q_1,\dots, Q_s)$ is \emph{almost pairwise independent}.}
        \label{main-fig:large-independent-set}
    \end{subfigure}
    \hfill
    \begin{subfigure}[t]{0.45\textwidth}
        \centering
        \begin{tikzpicture}[
    vertex/.style={circle, draw, inner sep=1pt, minimum size=12pt},
    edge label/.style={midway, above, sloped, fill=white, inner sep=1pt, font=\small, draw=none}
]
    \node[vertex] (Qstar) at (0,0) {$Q^*$};

    \node[vertex] (Q1) at (4, 1.8) {$Q_1$};
    \node[vertex] (Q2) at (4, 0.0) {$Q_2$};
    \node[draw=none] (dots) at (4,-0.8) {$\vdots$};
    \node[vertex] (Qt) at (4, -1.8) {$Q_{t}$};

    \draw (Qstar) -- node[edge label] {\scriptsize $\rank(Q^*+Q_1)\le r$} (Q1);
    \draw (Qstar) -- node[edge label] {\scriptsize $\rank(Q^*+Q_2)\le r$} (Q2);
    \draw (Qstar) -- node[edge label] {\scriptsize $\rank(Q^*+Q_{t})\le r$} (Qt);
\end{tikzpicture}
        \caption{\textbf{Large degree vertex. }Each polynomial in the tuple $(Q_1,\dots, Q_t)$ can be expressed in a new basis $Q_j = Q^* + \Gamma_j(L_{j1},\dots, L_{jr})$, which is \emph{essentially} low-rank up to the common polynomial $Q^*$. Intuitively, $Q^*$ explains the (potentially) high-rank part of many polynomials simultaneously.} \label{main-fig:large-deg-vertex}
    \end{subfigure}
    \caption{A win-win scenario for sampling lower bound.}\label{main-fig:win-win}
\end{figure}
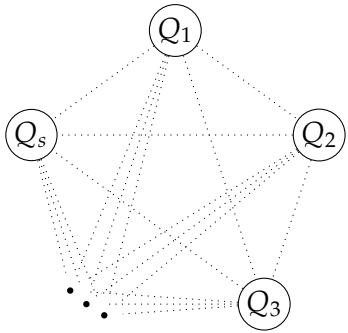
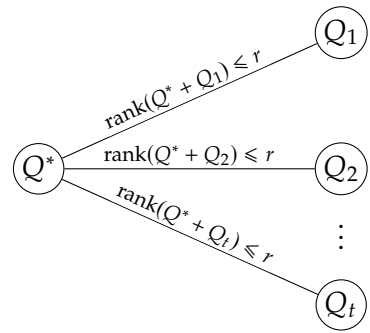

Suppose first that $\rank(Q_i+Q_j) \geq r$ for all $i \neq j \in S$. Then, the corresponding distribution $(Q_i)_{i \in S}$ is \emph{almost pairwise independent} with each coordinate having expectation bounded away from $1/3$.
Suppose for concreteness that at least $s/2$ of them have the expectation at least $1/3+\delta$.%
\footnote{Otherwise, at least $s/2$ of them have the expectation at most $1/3-\delta$, and the same argument applies.}
In particular, on average at least $1/3+\delta$ fraction of these polynomials output $1$, and by pairwise almost independence, a standard Chebyshev-like argument implies that $(Q_i)_{i \in S}$ is $1-o_s(1)$-far from $\Ber(1/3)^{\ox s}$.

For the second case for each $j \in S$ we can write $Q_j = Q_{i^*} + \Gamma_j(L_{j1},\dots,L_{jr})$ for some function $\Gamma_j \colon \F_2^r \to \F_2$, some linear functions $L_{j1},\dots,L_{jr}$, and the \emph{common} quadratic polynomial $Q_{i^*}$. That almost reduces it to the setting of \cref{main-thm:bounded-linear-rank}, except for the common part $Q_{i^*}$. We overcome this slight complication by taking the union bound over the conditioning on the two possible output values of $Q_{i^*}$, which effectively corresponds to restricting the domain of the polynomials. See \cref{lem:dist-tv-common-Q} for details.

By choosing the appropriate values of $r$ and $s$, we obtain the following result.

\begin{stheorem}[Quadratic distributions; formal version in \cref{thm:quadratic-sources}]\label{main-thm:quadratic}
    Let $\cQ = (Q_1,\dots, Q_n)$ be a distribution, where each $Q_i\colon\F_2^m\to\F_2$ is a quadratic polynomial.
    Then $\norm{\cQ - \Ber(1/3)^{\ox n} }_\tv\geq 1 - n^{-c/\log\log(n)}$ for some absolute constant $c>0$.
\end{stheorem}

\medskip

The next natural step is to show a lower bound for bounded $\rank_2$ distributions. Specifically, we prove the following theorem.

\begin{stheorem}[Bounded $\rank_2$ distributions; formal version in \cref{thm:bounded-rank-2}]\label{main-thm:bounded-quadratic-rank}
    Let $\cP = (P_1,\dots,P_n)$ be a distribution, where each $P_i$ satisfies $\rank_2(P_i) \leq r$, i.e., $P_i=\Gamma_i(Q_{i1},\dots, Q_{ir})$ for some arbitrary function $\Gamma_i$ and degree-$2$ polynomials $Q_{i1},\dots, Q_{ir}$.
    Then $\norm{\cP - \Ber(1/3)^{\ox n}}_\tv \geq 1 - \exp(-\log(n)^{c/r^2})$ for some absolute constant $c>0$.
\end{stheorem}

Below we describe the ideas that go into the proof of the theorem. The main idea is to use the technique known as \emph{regularization} in higher order Fourier analysis \cite{green2007distribution,hatami2019higher}. We describe the special case of it for $P_i$'s of bounded $\rank_2$. Later, we will generalize it to arbitrary degree $d$.

Suppose we have a collection of quadratic polynomials $(Q_1,\dots, Q_r)$. 
We define the structural notion of \emph{regularity} as follows. A collection or \emph{factor} of polynomials $(Q_1,\dots, Q_r)$ is said to be $k$-regular if all nonzero linear combinations $\tilde Q = \sum_i\lambda_i Q_i$ have rank greater than $k$.
By the the structure of quadratic polynomials (Dickson's Lemma~\cref{thm:structure-of-deg2}),
if the factor $(Q_1,\dots, Q_r)$ is $k$-regular,
then all nonzero linear combinations $\sum_i\lambda_i Q_i$ are almost balanced in the sense that $\Pr_X[\sum_i\lambda_i Q_i(X) = 1] = (1 \pm 2^{-k/2})/2$.
A standard Fourier-analytic argument shows that if $k \gg r$, then all nonzero linear combinations $\sum_i\lambda_i Q_i$ are almost balanced, and the distribution of $(Q_1,\dots, Q_r)$ must be close to uniform.

Suppose now that $P$ has $\rank_2(P)\leq r$, and hence can be written as $P = \Gamma(Q_1,\dots, Q_r)$ for some quadratic polynomials $Q_i$. Observe that if the factor of quadratics $(Q_1,\dots, Q_r)$ is not $k$-regular, then there exists some linear combination $\sum_i\lambda_i Q_i$ whose rank is at most $k$, and hence we can remove one of the $Q_i$'s, and replace it with at most $k$ linear functions, and hence $P$ can be written as $P = \Gamma(Q_1,\dots, Q_{r-1}, L_1,\dots,L_k)$, where all $L_i$'s are linearly independent. Continuing this process, we get a \emph{refinement} $(Q_1,\dots, Q_{r'}, L_1,\dots,L_s)$ consisting of $r' \leq r$ quadratic polynomials and $s' \leq kr$ linear function such that $P = \Gamma(Q_1,\dots, Q_{r'}; L_1,\dots,L_{s'})$ and $(Q_1,\dots, Q_{r'}, L_1,\dots,L_{s'})$ is $k$-regular.

\medskip

We can now return to the proof overview of the lower bound for distributions of bounded $\rank_2$. The proof uses induction. Let $\cP = (P_1,\dots, P_n)$ be such a distribution, where each $P_i$ can be written as a function of $r$ polynomials of degree 2, which we write as $P_i=\Gamma_i(Q_{i1},\dots, Q_{ir})$. Let us write $\cQ_i$ as a shorthand for $(Q_{i1},\dots, Q_{ir})$.

We may assume that all $\cQ_i$'s are $k$-regular for some sufficiently large $k$, as otherwise, we can perform a refinement step; this step will reduce the number of actual quadratics and might introduce new linear functions into $\cQ_i$. Now, we construct a graph, similar to \cref{main-fig:win-win}, where the vertices correspond to  $\cQ_i$'s. We put an edge between $\cQ_i$ and $\cQ_j$ if the pair $(\cQ_i; \cQ_j)$ is not $k$-regular for some parameter $k$, i.e. if the $2r$-tuple
$\cQ_{i,j} = (Q_{i1}, \dots, Q_{ir}; Q_{j1},\dots, Q_{jr})$ is not $k$-regular.

Once again, the win-win scenario of \cref{main-fig:win-win} is applicable.
If a large number $S \seq [n]$ satisfies  the pairwise regularity condition, namely, $\cQ_{i,j} = (Q_{i1}, \dots, Q_{ir}; Q_{j1},\dots, Q_{jr})$ is $k$-regular for all $i\neq j \in S$, then a Chebyshev-like argument implies that the corresponding distribution if $1-o(1)$-far from $\Ber(1/3)^n$.

Otherwise, there exists some $\cQ_{i^*}$ and a large set $S \seq [n]$ such that for all $j \in S$ the factor $\cQ_{i^*,j}$ is not regular.
We show that in this case, we can change the basis of $\cQ_j$, possibly introducing a few linear functions along the way, so that for all $j \in S$  the function $P_j$ can now be written in the form $P_j = \Gamma'_j(Q^*; Q_{j1},\dots, Q_{j, r-1}; \cL_j)$ where $\cL_j$ is a collection of linear functions, and $Q^*$ is a quadratic polynomial common to all $j \in S$. This step generalizes \cref{main-fig:large-deg-vertex} from a single polynomial to factors of polynomials. In this new basis, the polynomials $P_1,\dots,P_t$ have \emph{essentially} had their quadratic rank reduced to $r-1$, which allows us to invoke the induction hypothesis after handling $Q^*$; see the proof of \cref{thm:bounded-rank-2} for the details.

\subsection{Lower bounds against polynomial distributions of degree greater than 2}

In the next step we lift \cref{main-thm:bounded-quadratic-rank} to a lower bound for cubic distributions by an argument analogous to the one used to derive the lower bound for quadratics from bounded $\rank_1$. We obtain the following concrete bound for degree three.

\begin{stheorem}[Cubic distributions; formal version in \cref{thm:cubics}]\label{main-thm:cubic}
    Let $\cC = (C_1,\dots,C_n)$ be a distribution, where each $C_i\colon\F_2^m\to\F_2$ is a polynomial of degree $3$. Then $\norm{\cC - \Ber(1/3)^{\ox n}}_\tv \geq 1 - \exp(-c\sqrt{\log\log(n)})$ for some absolute constant $c>0$.
\end{stheorem}

In principle, one could continue in this manner and obtain concrete bounds for higher-degree distributions, but at this point deriving an explicit lower bound $1-\eps_d(n)$ requires a lengthy calculation, while the gain over the simpler bound $1-o_n(1)$ is comparatively modest. We therefore content ourselves with this form of the result and turn to the final conclusion of the paper: a general sampling lower bound for degree-$d$ distributions.

\begin{stheorem}[Degree-$d$ distributions; formal version in \cref{thm:deg-d-lower-bound}] \label{main-thm:degree-d}
    Let $\cP = (P_1,\dots, P_n)$ be a distribution, where each $P_i\colon\F_2^m\to\F_2$ is a polynomial of degree $d$. Then $\norm{\cP - \Ber(1/3)^{\ox n}}_\tv \geq 1 - o_{d}(1)$.
\end{stheorem}

We now sketch the main idea behind the proof of this \cref{main-thm:degree-d}. In order to describe the proof we first need the following lemma, which we call the \emph{Chebyshev lemma}.%
\footnote{The name is due to Chebyshev's inequality playing the main role in its proof.}

\subsection{Chebyshev lemma}

Before proceeding to this lemma, we recall the definition of degree-$d$ rank. The $\rank_d(P)$ of a function $P\colon\F_2^m\to\F_2$ is the smallest positive integer $k$ for which $k$ polynomials $Q_1,\dots, Q_k$ of degree $d$ exist so that $P$ can be written as $P=\Gamma(Q_1,\dots, Q_k)$ for some $\Gamma\colon\F_2^k\to \F_2$. When we speak of the rank of a polynomial of degree exactly $d+1$, we mean its $\rank_d$.

\begin{slemma}[Chebyshev lemma; formal version in \cref{lem:chebyshev-bounded-rank-d}]\label{main-lem:chebyshev-lemma}
    Let $\cP = (P_1,\dots, P_n)$ be such that each $P_i$ has bounded $\rank_d$; i.e., $P_i=\Gamma_i(Q_{i1},\dots, Q_{ir})$ for some $\Gamma_i$ and degree-$d$ polynomials $Q_{i1},\dots, Q_{ir}$. Suppose further that for all $i\neq j$ the distribution defined by the tuple
    \begin{equation}\label{main-eq:pairwise-uniform}
    \cQ_{ij} = (Q_{i1}, \dots, Q_{ir}; Q_{j1},\dots, Q_{jr}) \qquad \text{ is $2^{-r}$-close to uniform.}\tag{Pairwise Uniform}
    \end{equation}
    Then, $\norm{\cP - \Ber(1/3)^{\ox n}}_\tv \ge 1 - 2^{-r} - o(1)$.
\end{slemma}

The requirement of being close to uniform can be replaced by the structural requirement of \emph{regularity}. A collection or \emph{factor} of degree-$d$ polynomials $(Q_1,\dots, Q_r)$ is $k$-regular if all nonzero linear combinations $\tilde Q = \sum_i\lambda_iQ_i$ have rank more than $k$. The celebrated rank-bias theorem of \cite{green2007distribution,kaufman2008worst} gives a tradeoff between $k$ and the bias of $\tilde Q$. More precisely, if $k$ is large enough as a function of $\delta$ and $d$, then $\Pr_X[\tilde Q(X)=1]=(1\pm \delta)/2$.
Hence, if the factor of polynomials $(Q_1,\dots, Q_r)$ is $k$-regular, then all nonzero linear combinations of are balanced and the joint distribution $(Q_1,\dots, Q_r)$ is close to uniform by a standard Fourier-analytic argument.
It follows that for some $k=k(d,\delta, r)$ the following implies \cref{main-eq:pairwise-uniform}.
    \begin{equation}\label{main-eq:pairwise-regularity}
    \cQ_{ij} = (Q_{i1}, \dots, Q_{ir}; Q_{j1},\dots, Q_{jr}) \qquad \text{ is } k\text{-regular}.\tag{Pairwise Regularity} 
    \end{equation}

Loosely speaking, our Chebyshev lemma enables us to prove a sampling lower bound by applying induction on the $\rank_d$. We explained this idea to some extent in our lower bound for bounded quadratic rank. If we have a distribution $\cP = (P_1,\dots, P_n)$ such that each $P_i$ has bounded $\rank_d$, then either
\begin{enumerate}[(a)]
    \item there exists a large subset of coordinates satisfying \cref{main-eq:pairwise-regularity}, which puts us in the setting of \cref{main-lem:chebyshev-lemma}; or
    \item  there exists a low-rank structure, which allows us to reduce a certain parameter (roughly corresponding to $r$), and hence apply the induction hypothesis.
\end{enumerate}

The main tool that makes the Chebyshev lemma useful is \emph{regularization}. We have already discussed regularization in part when considering distributions of bounded quadratic rank. In full generality, regularizing a factor $\cQ=(Q_1,\dots,Q_r)$ of degree-$d$ polynomials $Q_i\colon\F_2^m\to\F_2$, is the process of producing another factor of degree-$d$ polynomials $\cQ'=(Q'_1,\dots,Q'_{r'})$ that enjoys the following properties. 

\begin{itemize}
    \item \textit{(Regularity)} Each linear combination $\sum_i\lambda_i Q'_i$ has high rank.
    \item \textit{(Refinement)} Everything that was computable by $\cQ$, is also computable by $\cQ'$. Formally, if $P = \Gamma(Q_1,\dots, Q_r)$ for some $\Gamma$, then $P=\Gamma'(Q'_1,\dots, Q'_{r'})$ for some $\Gamma'$.
    \item \textit{(Small size blow-up)} The number of polynomials in $\cQ'$ is not much larger than $\cQ$. More precisely, $r'$ depends only on $r, d,$ and the regularity parameter $k$, and it is independent of the number of variables $m$.
\end{itemize}

We now sketch the main idea behind the proof of \cref{main-thm:degree-d}. The proof will be similar to \cref{main-thm:bounded-quadratic-rank} for distributions with bounded quadratic rank. Roughly speaking, the proof of \cref{main-thm:degree-d} works by formulating a generalized version of the idea depicted in \cref{main-fig:win-win} as a separate structural theorem, which we call \emph{sunflower pairwise regularization}. This structural result is of independent interest to us, and may be useful beyond the present argument.

\begin{stheorem}[Sunflower pairwise regularization; formal version in \cref{thm:lossy-pairwise-regularization}]\label{main-thm:sunflower-pairwise-regularization}
    Let $(P_1,\dots, P_n)$ be a collection of degree-$d$ polynomials $P_i\colon\F_2^m\to\F_2$. Then there exists a subset $S \seq [n]$ of polynomials of size $w=n^{\Omega(1)}$ such that for each $i \in S$ the following holds.
    \begin{itemize}
        \item Every $P_i$ can be written in the form $P_i=\Gamma_i(Q_1^*,\dots, Q^*_{c}; Q_{i1},\dots, Q_{ir})$ where $Q_1^*,\dots, Q^*_c$ appear in all of the $P_i$ and $Q_{i1},\dots, Q_{ir}$ is unique to $P_i$. That is, the sets $\cG_i=\{Q_1^*,\dots, Q^*_{c}, Q_{i1},\dots, Q_{ir}\}$ form a sunflower with the core $\cQ^*=(Q^*_1,\dots, Q^*_c)$ and petals $\cQ_i=(Q_{i1},\dots, Q_{ir})$.
        \item The number $r+c$ is small; in particular, it is independent of $m$ and $n$.
        \item For all $i\neq j$ the tuple $\cQ_{ij} = (Q_{i1},\dots, Q_{ir};Q_{j1},\dots,Q_{jr})$ satisfies \cref{main-eq:pairwise-regularity}.
    \end{itemize}
\end{stheorem}

With \cref{main-thm:sunflower-pairwise-regularization} in hand, the only remaining issue is handling the common polynomials $Q^*_1,\dots, Q^*_c$ without which we can apply our Chebyshev lemma (\cref{main-lem:chebyshev-lemma}) to the pairwise regular set $S$. In a separate lemma we prove that because $c \ll w$, the effect of $Q^*_i$'s is not significant (see \cref{lem:dist-tv-common-Q}), which concludes the proof of our lower bound for degree-$d$ distributions.

We close our discussion of the results by turning to the proof of \cref{main-thm:sunflower-pairwise-regularization}. The goal in this theorem is to write $P_i$ in the form $P_i=\Gamma_i(\cQ^*; \cQ_i)$ satisfying the three conditions. At the beginning, such a representation exists trivially by letting $\cQ^*$ to be empty and $\cQ_i$ to be simply $(P_i)$; however, \cref{main-eq:pairwise-regularity} is not necessarily satisfied. At this point, we construct a graph similar to \cref{main-fig:win-win} connecting $\cQ_i$ to $\cQ_j$ whenever the pair $(\cQ_i; \cQ_j)$ is not regular enough. Then if we find a large independent set in this graph (\cref{main-fig:large-independent-set}) we simply stop and announce the corresponding independent set as our final subset of polynomials. Otherwise, a vertex of large degree exists as in \cref{main-fig:large-deg-vertex}, which allows for a change of basis, that introduces a new common polynomial into the core $\cQ^*$ at the expense of introducing a few polynomials into the petals $\cQ_i$. We repeat this step until we get a subset of polynomials that are pairwise regular. To ensure that $\cQ_i$ are individually regular, we need to interweave such steps with usual refinement/regularization steps. The main concern is why this process must terminate. The reason is that at each step, a polynomial is replaced by a few polynomials of strictly smaller degree, so the process can continue only for finitely many steps.

\subsection{Open problems}\label{sec:open-problems}
Below are some problems that arise naturally from this work.
\begin{itemize}
    \item Improve the lower bounds on the statistical distance. In particular, is it true that 
    $\norm{\cP - \Ber(1/3)^{\ox n}}_\tv \geq 1 - \exp(-\Omega_d(n))$
    for all $d$?
    
    \item What about samplers defined by polynomials of degree $d$, where $d$ grows as a function of $n$, say $d = c\log(n)$. The results of this work seem to be useful only when $d$ is a \emph{very} slowly growing functions of $n$.

    \item Characterize all achievable values of $\Pr_X[P(X)=1]$, where $P\colon\F_2^m\to\F_2$ is a polynomial of degree $d$ and $m$ is allowed to be any positive integer. A related problem is the weight distribution of binary Reed-Muller codes; see \cite{kaufman2012weight,kasami1976weight,kasami1970weight}.
    
    \item For degree $d=3$ how close can $\Pr_X[P(X) = 1]$ be to $1/3$? Below we provide several examples. First consider the following polynomials with $\Pr_X[P(X) = 1] < 1/3$.
    \begin{itemize}
        \item The simplest example is $P(x)=x_1x_2$ with $\Pr[P(x)=1] =1/4$. The gap is $1/3-1/4=1/12$.
        \item A more interesting example is the polynomial $P(x) =x_1x_2+x_3x_4x_5$. The first term gives $\Pr[x_1x_2=1] = 1/4$, and the second term gives $\Pr[x_3x_4x_5=1] = 1/8$. Hence $\Pr[P(x)=1] = 1/4+1/8-2 \cdot (1/4)\cdot (1/8) = 5/16$. The gap is $1/3-5/16=1/48 \approx 0.0208$.
    \end{itemize}
    Next, we describe several examples with $\Pr_X[P(X) = 1] > 1/3$.
    \begin{itemize}
        \item The polynomial $P(x) = x_1 \cdot (1-x_2x_3)$ satisfies $\Pr_X[P(X) = 1] = 1/2 \cdot 3/4 = 3/8$. The gap is $3/8-1/3 = 1/24 \approx 0.0417$.
        \item The polynomial $P(x) = x_1x_2 + x_3(x_4x_5+x_6x_7)$. The first term $\Pr[x_1x_2=1] = 1/4$, and the second term gives $\Pr[x_3(x_4x_5+x_6x_7)=1] = 3/16$. Hence $\Pr[P(x)=1] = 1/4+3/16-2 \cdot (1/4) \cdot (3/16) = 11/326$. The gap is $11/32-1/3 = 1/96 \approx 0.0104$.
        \item Our last example is $P(x) = x_1x_2x_3 + x_1x_4 + x_4x_5x_6 + x_7x_8x_9$.
        In this example, we have $\Pr[P(x)=1] = 43/128$, and the gap is $43/128-1/3=1/384 \approx 0.0026$.
    \end{itemize}
    It would be interesting to better understand such examples, and compare the gaps we obtain here to the lower bound $\delta_3$ from \cref{lemma:deg-3-dyadic}. 
    
    We are grateful to Mark Kahn for providing the examples above.
    The examples were essentially achieved using (an LLM tool which ran) a random walk $(P_0,P_1,P_2,\dots)$ on cubic polynomials, starting with, say, $P_0(x) = 0$, and generating $P_{i+1}$ by adding or removing a random monomial from $P_i$. In each step we compute $\Pr[P_i(x)=1]$ and check how close it is to $1/3$.
\end{itemize}

\section{Preliminaries}\label{sec:prelims}

\subsection{Basic definitions and notation}\label{sec:notation-and-definitions}

\paragraph{Polynomials.} We work with multilinear polynomials $P\in\F_2[x_1,\dots,x_m]$ of degree $d$, viewed as functions $P\colon\F_2^m\to\F_2$. It is understood that a degree-$d$ polynomial means a polynomial of degree \emph{at most} $d$ unless otherwise stated. 

The definition of polynomials can be extended to functions of the form $P\colon V\to\F_2$ where $V$ is an $m$-dimensional vector space over $\F_2$. We use the fact that $V$ is isomorphic to $\F_2^m$, so after fixing some linear isomorphism $\phi\colon \F^m_2\to V$ we say $P$ is a polynomial of degree $d$ if the function $P'\colon \F_2^m\to\F_2$ given by $P'(x)=P(\phi(x))$ is a polynomial of degree $d$. This can be further extended in the natural way to functions $Q\colon W+h\to\F_2$ where $W+h$ is an affine subspace of $V$. In this case, $Q$ is defined to be a polynomial of degree $d$ if $Q'(x) = Q(\alpha(x) + h)$ is a polynomial of degree $d$, where $\alpha\colon \F_2^{\dim(W)}\to W$ is a linear isomorphism. This definition is independent of the particular choice of isomorphisms $\phi$ and $\alpha$. The class of polynomials is closed under affine restrictions.

We now define our main object of study, polynomial distributions, sometimes also referred to as polynomial sources in the context of randomness extractors.

\begin{definition}[Degree-$d$ polynomial distribution]
    A distribution $\cP$ over $\Bits^n$ is a \defemph{degree-$d$ polynomial distribution} if for some $m\in\N$ there exist $n$ polynomials $P_1,\dots, P_n\colon \F_2^m\to \F_2$ each of degree $d$ such that the tuple sampled according to $(P_1(X),\dots, P_n(X))$ where $X$ is drawn uniformly at random from $\F_2^m$ has the same distribution as $\cP$. We assert this by writing $\cP=(P_1,\dots, P_n)$ usually suppressing the common random input $X$.
\end{definition}

Often, we will identify $P$ with the distribution $P(X)$ it defines on a random input $X$, and so we treat $P$ both as a polynomial and a random variable.

The \defemph{bias} of a polynomial $P$ is the quantity defined as $\bias(P) \defeq \abs{\E_{X}[(-1)^{P(X)}]}$. This quantity without the absolute values is referred to as the \defemph{signed bias} $\bias_\pm(P)$.

\paragraph{Distributions. } Bernoulli distribution $\Ber(\rho)$ with parameter $\rho$ will be more compactly written as $\cB_\rho$. The uniform distribution on $\Bits$ is denoted by $\cU$. The total variation distance between two distributions $\cD_1,\cD_2$ defined over the same countable domain $\Omega$ is $\norm{\cD_1 - \cD_2}_\tv\defeq \half \sum_{x\in\Omega}\abs{\cD_1(x)-\cD_2(x)}$. Equivalently, $\norm{\cD_1 - \cD_2} = \max_{A\seq \Omega}\big(\cD_1(A) - \cD_2(A)\big)$, where $\cD_i(A)$ is the probability that the output of $\cD_i$ belongs to $A$. If $\cT$ is a distribution over a set $S$, we denote by $\cT^{\ox n}$ the distribution over $S^n$ obtained by $n$ independent copies of $\cT$. Given a collection of distributions $\{\cD_i : i \in I\}$ over a common domain, and a distribution $\cC$ over the index set $I$, their convex combination (weighted by $\cC$) is denoted by $\E_{i \sim \cC}[\cD_i]$; a sample $X$ from $\E_{i \sim \cC}[\cD_i]$ is obtained by first sampling an index $i\sim \cC$ and the sampling $X\sim \cD_i$.

\paragraph{Other conventions. } We use the standard notation $[n]$ to refer to $\{1,\dots, n\}$. We include zero in the set of natural numbers, so $\N=\{0,1,\dots\}$. Throughout, $\log$ denotes the base-$2$ logarithm, and $\exp(x)$ denotes $2^x$. For disjoint sets $A$ and $B$, the union is sometimes written as $A \sqcup B$ to highlight disjointness. For real numbers $a,b,c$, whenever we write $a = b \pm c$, that means there exists some $d\in[-c,c]$ that satisfies the equality $a=b+d$. For a function $f\colon A\to B$ its \defemph{restriction} to $A'\seq A$ is the function $f\rest{A'} \colon  A'\to B$ given by $f\rest{A'}(x)=f(x)$. A \defemph{dyadic number} is a rational number $q$ that can be written in the form $a/2^b$ for some $a\in \mathbb{Z}$ and $b\in \mathbb{N}$. We call $b$ the granularity of this dyadic representation.

\subsection{Some properties of total variation distance}\label{sec:properties-of-tv-dist}

A lemma we use extensively throughout the paper is the following XOR lemma, commonly attributed to Vazirani, which states that an $\eps$-biased distribution is close to uniform whenever $\eps$ is exponentially small in the output length.

\begin{lemma}[Vazirani's XOR lemma]\label{lem:vazirani-xor}
    Let $\cX=(X_1,\dots, X_n)$ be a distribution over $\F_2^n$. If for all nonzero linear functions $L\colon \F_2^n\to\F_2$ it holds that 
    \[
    \Pr[L(X_1,\dots, X_n) = 0] = \frac{1}{2} \pm \frac{\eps}{2},
    \]
    then, $\norm{\cX - \cU^{\ox n}}_\tv\le 2^{n/2}\eps$. \hfill (Proof in \cref{sec:deferred-proofs}.)
\end{lemma}

The following lemma may be regarded as union bound for the distances between distributions.
\begin{lemma}[Distance from convex combination]\label{claim:tv-dist-avg}\label{lem:dist-of-convex-comb}
    Let $\cD, \cT$ be two distributions over a finite domain $X$, and let $\cT=\sum_ic_i\cT_i$ be the convex combination of distributions $\cT_1,\dots, \cT_k$ according to the weights $c_1,\dots, c_k$ where $c_i>0$ and $\sum_ic_i = 1$. Suppose that $\norm{\cD - \cT_i}_\tv = 1 - \eps_i$ for some $\eps_i\in [0, 1/2]$. Then, 
    \begin{equation*}
        \norm{\cD - \cT}_\tv \geq 1 - \sum_{i=1}^k (1+c_i) \eps_i.
    \end{equation*}
\end{lemma}

\begin{proof}
    Let $A_i = \{x \in X : \cD(x) > \cT_i(x)\}$. Then we have $\norm{\cD - \cT_i}_\tv = \cD(A_i) - \cT_i(A_i)$.
    In particular $\cD(A_i) \geq 1-\eps_i$ and $\cT_i(A_i) \leq \eps_i$.
    Consider the set $A = \bigcap_{i=1}^k A_i$, and note that 
    \[
    \cD(A) \geq 1 - \sum_{i=1}^k \cD(X \setminus A_i) \geq  1-\sum_{i=1}^k \eps_i.
    \]
    On the other hand,
    \[
    \cT(A) = \sum_{i=1}^k c_i \cT_i(A) \leq \sum_{i=1}^k c_i \cT_i(A_i)  \leq \sum_{i=1}^k c_i  \eps_i.
    \]
    Therefore,
    \begin{equation*}
        \cD(A) - \cT(A) \geq \Bigl(1-\sum_{i=1}^k \eps_i\Bigr) - \sum_{i=1}^k c_i \eps_i
        = 1 - \sum_{i=1}^k (1+c_i) \eps_i,
    \end{equation*}
    as required.
\end{proof}

\begin{lemma}[Distance between product distributions \cite{horacsek2025sourcedecoding}]\label{lem:dist-of-product-distributions}
    Let $\cD=\cD_1\ox \cdots\ox \cD_n$ and $\cT=\cT_1\ox \cdots\ox \cT_n$ be product distributions over $\Bits^n$ where all the marginals $i\in [n]$ satisfy $\norm{\cD_i-\cT_i}_\tv\ge \delta$. Then $\norm{\cD-\cT}_\tv \ge 1- 2e^{-\delta^2n/12}$.
\end{lemma}

\begin{lemma}[Conditioning]\label{lem:dist-conditioning}
    Fix a mapping $f \colon \Bits^m \to \Bits^n$.
    Let $\cU^{\ox m}$ be the uniform distribution over $\Bits^m$ and let $\cY$ be a distribution over $\Bits^n$ such that $\norm{f(\cU^{\ox m})-\cY}_\tv \geq 1-\eps$.
    Let $\cU_S$ be the uniform distribution over $S  \seq \Bits^m$ of size $\abs{S} = \delta \cdot 2^m$.
    Then $\norm{f(\cU_S)-\cY}_\tv \ge 1-\eps(1+1/\delta)$.
\end{lemma}
\begin{proof}

    Let $A \seq \Bits^n$ be an event such that $\Pr[f(\cU^{\ox m}) \in A] - \Pr[\cY \in A] = \norm{f(\cU^{\ox m})-\cY}_\tv \geq 1-\eps$. In particular, $\Pr[f(\cU^{\ox m}) \in A] \geq 1-\eps$ and $\Pr[\cY \in A] \leq \eps$.    
    By the assumption on $S$, we have $\Pr[f(\cU_S) \notin A] \leq \Pr[f(\cU^{\ox m}) \notin A]/\delta$.
    \begin{eqnarray*}
        \norm{f(\cU_S)-\cY}_\tv
         & \geq & \Pr[f(\cU_S) \in A] - \Pr[\cY \in A] \\
         & = & 1 - \Pr[f(\cU_S) \notin A] - \Pr[\cY \in A] \\
         & \geq & 1 - \Pr[f(\cU^{\ox m}) \notin A]/\delta - \eps  \\
         & \geq & 1 - \eps/\delta - \eps \\
         & = & 1 - \eps\cdot(1+1/\delta),
    \end{eqnarray*}
    as required.
\end{proof}

\begin{lemma}[Fixing coordinates]\label{lem:dist-tv-common-Q}
    Let $Q_1,\dots,Q_k \colon \Bits^m \to \Bits$ and $Q^* \colon \Bits^m \to S$ be $k+1$ functions, and fix $\Gamma \colon \Bits^k \times S \to \Bits^n$.
    For $\sigma \in S$ define $\Gamma_\sigma \colon \Bits^k \to \Bits^n$ to be the restriction of $\Gamma$ by setting the last variable to be $\sigma$,
    that is $\Gamma_\sigma(x_1,\dots,x_k) = \Gamma(x_1,\dots,x_k,\sigma)$.
    
Let $X$ be the uniform distribution over $\Bits^m$, and let $\cY$ be a distribution over $\Bits^n$ such that for all $\sigma \in S$ we have $\norm{\Gamma_\sigma(Q_1(X),\dots,Q_k(X))-\cY}_\tv \geq 1-\eps$.          
    Suppose that $\Pr[Q^*(X) = \sigma] \geq \tau$ for all $\sigma \in S$.
    Then
    \begin{equation*}
        \norm{\Gamma(Q_1(X),\dots,Q_k(X),Q^*(X))-\cY}_\tv \geq 1-\frac{4\eps  \abs{S}}{\tau}.
    \end{equation*}
\end{lemma}
\begin{proof}
    Note that
    \begin{equation*}
        \Gamma(Q_1(x),\dots,Q_k(x),Q^*(x))
        = \sum_{\sigma \in S}\indicator{Q^*(x) = \sigma} \cdot \Gamma_\sigma(Q_1(x),\dots,Q_k(x)).
    \end{equation*}
    By \cref{lem:dist-conditioning} if we condition on $Q^*(X) = \sigma$, then
    \[
    \norm{\Gamma_\sigma(Q_1(X),\dots,Q_k(X)) \rest{Q^*(X) = \sigma} - \cY}_\tv \ge 1-\eps(1+1/\tau).
    \]
    Since $\Gamma(Q_1(X),\dots,Q_k(X),Q^*(X))$ is the weighted average of $\Gamma_\sigma(Q_1(X),\dots,Q_k(X))$, it follows by \cref{lem:dist-of-convex-comb} that
    \begin{eqnarray*}
        \norm{\Gamma(Q_1(X),\dots,Q_k(X),Q^*(X))-\cY}_\tv
        & \geq & 1-\sum_{\sigma \in S}(1+\Pr[Q^*(X)=\sigma]) \cdot \eps(1+1/\tau) \\
        & = & 1 - (\abs{S}+1)\eps(1+1/\tau) \\
        & \geq & 1 - 4\eps \abs{S}/\tau,
    \end{eqnarray*}
    as required.    
\end{proof}

\subsection{Rank, regularity, and factors of polynomials}\label{sec:factors}

Next, we introduce the standard notions relating polynomials which are also found in previous work \cite{green2007distribution,kaufman2008worst, bhattacharyya2014algorithmic}. 

\begin{definition}[Rank \cite{green2007distribution, kaufman2008worst}]
    Let $P\colon \F^n\to \F$ be a function, not identically zero. The degree $d$ rank of $P$, denoted by $\rank_d(P)$, is the smallest positive integer $k$ for which there exist $k$ degree $d$ polynomials $Q_1,\dots, Q_k$ and a function $\Gamma\colon \F^k\to\F$ such that
    \[
    P(x) = \Gamma(Q_1(x), \dots, Q_k(x)) \qquad \text{ for all } x.
    \]
\end{definition}

We may use the terms \emph{linear rank}, \emph{quadratic rank}, and \emph{cubic rank} to refer to $\rank_1(P),\rank_2(P)$ and $\rank_3(P)$ respectively. The most important notion of rank for a degree $d$ polynomial $P$ is $\rank_{d-1}(P)$.

As with degree, our convention is that a polynomial of rank $r$ means a polynomial of rank \emph{at most} $r$. We will state explicitly whenever an exact degree or rank is intended.

\begin{fact}
    Let $P\colon \F_2^m \to \F_2$ be a polynomial of degree $d$ and rank $r$. Then the restriction of $P$ to any affine subspace $V+h$ is again a polynomial of degree $d$ and rank $r$.
    In other words, upper bounds on degree and rank are preserved under affine restrictions.
\end{fact}

\begin{definition}[Factors]
    A \defemph{factor} $\cF$ of degree $d\in\N$ with \defemph{dimension vector} $\vec{M} = (M_1,\dots, M_d)\in\N^d$ is a tuple of polynomials $\cF = (P_1,\dots, P_K)$ where $K = M_1 + \dots + M_d$ and for exactly $M_\ell$ many $P_i$'s we have that $\deg(P_i)=\ell$. The number $K$ is called the \defemph{dimension} or \defemph{locality} of $\cF$ and is denoted by $\dim(\cF)$ or $\card{\cF}$.
\end{definition}

We equip $\N^d$ with the \defemph{inverse lexicographic ordering}: for $\vec N = (N_1,\dots, N_d)$ and $\vec M = (M_1,\dots, M_d)$, we declare $\vec N < \vec M$ if, at the largest index $j$ for which $N_j \neq M_j$, one has $N_j < M_j$.

\begin{definition}[Refinement]
    Let $\cF = (P_1, \dots, P_K)$ be a factor. We say a factor $\cG = (Q_1, \dots, Q_{K'})$ is a \defemph{refinement} of $\cF$ (or that $\cG$ \defemph{refines} $\cF$) if there exists some $\Gamma\colon \F_2^{K'}\to \F_2^K$ such that
    \[
    (P_1,\dots, P_K) = \Gamma(Q_1,\dots, Q_{K'}).
    \]
\end{definition}

\begin{definition}[Rank \& regularity of a factor]
A factor $\cF = (P_1,\dots, P_K)$ of degree $d$ is said to be \defemph{$r$-regular} if for all nonzero linear combinations of $P_i$'s it holds that
\[
\rank_{\ell - 1}(\lambda_1P_1 + \dots + \lambda_K P_K) > r,
\]
where $\ell = \max_{i} \deg(\lambda_iP_i)$ is the \defemph{degree of the linear combination}.  Given a function $f\colon \N\to\N$, the factor $\cF$ is \defemph{$f$-regular} if it is $f(K)$-regular.

The \defemph{rank of the factor} $\cF$, denoted by $\rank(\cF)$, is the maximal number $r$ for which $\cF$ is $r$-regular.
\end{definition}

Given a factor $\cF$ that is free of linear dependencies, for $P\in\Span(\cF)$ we denote by \defemph{$\deg_\cF(P)$} the degree of the unique linear combination of polynomials in $\cF$ that defines $P$. Note that $\deg(P)\le \deg_\cF(P)$.

Sometimes, especially in \cref{sec:pairwise-reg}, we will view a factor $\cF$ as a set $\{P_1,\dots, P_K\}$. Accordingly, all the relevant definitions (e.g., dimension vector, refinement) then transfer naturally.

A \defemph{growth function} is a function $f\colon \N\to\N$ that is non-decreasing and satisfies $f(r)\ge r$ for all $r$.

\begin{definition}[A rapidly growing function]\label{def:rapid-function}
Let $f\colon \N\to\N$ be a growth function and let $d$ be a positive integer. Now define the function \defemph{$\psi_{d, f} \colon  \N^d\to\N$} recursively as follows. For the base case we have,
\[
\psi_{d,f}(M_1,0,\dots, 0) = M_1.
\]
Otherwise, if $M_i > 0$ for some $i\ge 2$, then let $\ell\ge 2$ be the smallest of such indices. Then,
\[
\psi_{d,f}(M_1,\dots,M_d)=\psi_{d,f}(M_1,\dots, M_{\ell-1}+f(\textstyle\sum_iM_i), M_\ell - 1, \dots, M_d).
\]
Now set \defemph{$\psi_{d, f}^* (K) = \psi_{d, f}(0, \dots, 0, K)$}.
\end{definition}

\begin{proposition}\label{prop:properties-of-rapid-function}
Let $\psi = \psi_{d, f}$ be defined as in \cref{def:rapid-function}. Then, $\psi(N_1,\dots,N_d) \le \psi(M_1,\dots, M_d)$ whenever $\sum_i N_i \le \sum_iM_i$ and $(N_1,\dots,N_d)\le (M_1,\dots, M_d)$.
\end{proposition}
The proof of this proposition is deferred to the \cref{sec:deferred-proofs} as it is a rather tedious induction and is not a primary focus of ours. This function serves as an upper bound on some parameters in regularization processes that appear here.

\begin{theorem}[Regularization \cite{green2007distribution}]\label{thm:classic-regularization}
    Every factor $\cF$ of degree $d$ admits an $f$-regular refinement $\cG$ with $\dim(\cG) \le \psi_{d, f}^*(\dim(\cF))$, where $f\colon \N\to\N$ is a growth function. 
\end{theorem}
\begin{proof}
    The proof is by strong induction on the dimension vector $\vec M = (M_1,\dots, M_d)$ of $\cF$, showing that for every $\cF$ there exists an $f$-regular refinement $\cG$ with $\dim (\cG) \le \psi_{d,f}(M_1,\dots, M_d)$. The bound $\dim(\cG) \le \psi^*_{d,f}(\dim(\cF))$ follows immediately from \cref{prop:properties-of-rapid-function}.
    
    As the base case, the empty factor is vacuously $f$-regular and there is nothing to do. Now suppose that $M_i > 0$ for some $i$, and that the claim holds for all smaller dimensions.

    We shall further assume that $\cF$ is free from linear dependencies, because otherwise we can discard some polynomials from $\cF$ and the resulting factor is smaller and still refines $\cF$. If $\cF$ is already $f$-regular, then we are done. Otherwise, identify some linear combination $R=\sum_i\lambda_i P_i$ such that $\rank_{\ell-1}(R)=r\le f(K)$, where $\ell=\deg_\cF(R)$.

    Then we replace a maximal degree $P_{i^*}$ with $\lambda_{i^*}\neq 0$ by $r$ many polynomials of degree at most $\ell - 1$ to obtain a new factor $\cF'$. Doing so would change the dimension vector from $\vec M = (M_1,\dots, M_d)$ to $\vec M' = (M'_1, \dots, M'_d)$ such that
    \[
    (M'_1,\dots, M'_d) < (M_1, \dots, M_{\ell - 1} + f(K), M_\ell - 1, \dots, M_d).
    \]
    We may now invoke the induction hypothesis on $\cF'$ to obtain an $f$-regular refinement $\cG$ of $\cF'$. The factor $\cG$ is also a refinement of $\cF$. Its dimension satisfies $\dim(\cG) \le \psi_{d, f}(\vec M')$ by induction. By \cref{prop:properties-of-rapid-function} it follows that $\psi_{d, f}(\vec M') \le \psi(\vec M)$.
\end{proof}

\subsection{The structure theorems and the bias-rank tradeoffs}\label{sec:structure-theorems}

Consider the beautiful theorem of bias implies low rank \cite{green2007distribution,kaufman2008worst}.

\begin{theorem}[Structure of Biased Polynomials \cite{kaufman2008worst}] \label{thm:kaufman-lovett}
    There exists a function \defemph{$c_\KL(d, \delta)$} such that for all polynomials $P$ over $\F_2$ of degree at most $d$ and $\bias(P)\ge \delta > 0$ it holds that $\rank_{d-1}(P) \le c_\KL(d, \delta)$.
\end{theorem}

Contrapositively, this theorem says that if $\rank_{d-1}(P) > c_\KL(d, \delta)$, then $\bias(P) < \delta$. We will use this fact extensively. This theorem is commonly known as \emph{bias implies low rank} and gives a tradeoff between bias and rank. It was first proved by Green and Tao in the case of $d \le \card{\F_p}$, but was later extended to all small prime fields by Kaufman and Lovett \cite{kaufman2008worst}.

For polynomials of degree $1,2,3$, explicit and efficient tradeoffs are known.
\begin{fact}[Structure of Degree $1$]
    Given a degree-$1$ polynomial $L$, if $\bias(L) > 0$ then $L$ is constant. In other words, $c_\KL(1, \delta) = 0$ for any positive $\delta$.
\end{fact}

\begin{theorem}[Dickson's lemma: Structure of quadratic polynomials {\cite[Theorem 6.30]{LidlNiederreiter1996}}]\label{thm:structure-of-deg2}
    Given a quadratic polynomial $Q$ there exists an even number $r$ and $L_1,\dots, L_{r+1}$ linear \emph{forms} where $Q$ can be written as
    \[
    Q = L_1L_2 + \dots + L_{r-1}L_r + L_{r+1} + c,
    \]
    where $c$ is a field constant, and $L_1,\dots, L_r$ are linearly independent.

    We have $\rank_1(Q) = r+1$ if $L_{r+1}$ is linearly independent from $L_1,\dots, L_r$, and $\rank_1(Q) = r$ otherwise. In the former case $\bias(Q) = 0$ and in the latter case $\bias(Q)=2^{-r/2}$. This implies that $c_\KL(2, \delta) = 2\log(1/\delta)$.
\end{theorem}

\begin{theorem}[Structure of biased cubic polynomials (Haramaty--Shpilka~\cite{haramaty2010structure})]\label{thm:haramaty-shpilka}
Let $C$ be a cubic polynomial with $\bias(C)\ge \delta>0$. Then there exist linear functions
$L_1,\dots,L_{r+r'}$, quadratic polynomials $Q_1,\dots,Q_r$, and a cubic polynomial $\Gamma$
(in $r'$ variables) such that
\[
C = \sum_{i=1}^{r} L_i Q_i \;+\; \Gamma(L_{r+1},\dots,L_{r+r'}),
\]
where $r=O(\log(1/\delta))$ and $r'=O(\log^4(1/\delta))$. This implies that $c_\KL(3, \delta)=O(\log^4(1/\delta))$.
\end{theorem}
For $d\ge 4$ the dependence on $\delta$ is very weak, so we leave it unspecified. We summarize this in the following table.
\[
\begin{array}{c|ccccc}
      & d=1 & d=2 & d=3 & d\ge 4 \\
\hline
c_\KL(d,\delta) & 0 & 2\log(1/\delta) & O(\log^4(1/\delta)) & O_{d, \delta}(1)  \\
\end{array}
\]

\subsection{Sunflowers for subspaces}

\begin{definition}
    A collection $S = \{V_1,\dots,V_s\}$ of (possibly repeated) subspaces in $\F^m$ is said to be a sunflower if for $C = \bigcap_{j=1}^s V_j$ it holds that $V_i \cap V_{i'} = C$ for all $i \neq i'$.

    The common intersection $C = \bigcap_{j=1}^s V_j$ is called the \emph{core} of the sunflower.
\end{definition}

\begin{lemma}[Sunflower Lemma for subspaces]\label{lemma:sunflower}
    Fix a finite field $\F = \F_p$ and $m > k \geq 1$.
    Let $S$ be a collection of linear subspaces of $\F^m$ of dimension $k$ (the collection might have repetitions).
    If $\abs{S} \geq s^{k+1}p^{\frac{k^2+k-2}{2}}$,
    then $S$ contains a sunflower of size at least $s$.
\end{lemma}

\begin{remark}
    The proof is a straightforward adaptation of the standard proof of the lemma for set systems without a linear structure.
    
    Note however that the bound on the size of $S$ is better than a naive application of the sunflower lemma, as our sets are of size $p^k$, and so the naive bound would be double exponential in $k$.
\end{remark}

\begin{proof}[Proof of \cref{lemma:sunflower}]
    The proof is by induction on $k$. For the base case of $k=1$ we either have at least $s$ repetitions of the same subspace or we have at least $\abs{S}/s \geq s$ different one dimensional subspaces, which all intersect only at the origin.

    For the induction step, let $k \geq 2$.
    Let $V_1,\dots,V_t$ be a maximal subset of $S$ whose pairwise intersection is only $\{0\}$.
    If $t \geq s$, we are done, as the sets form a sunflowers with the trivial core $C = \{0\}$.

    Otherwise, consider the set $X = \cup_{i=1}^t V_i \setminus \{0\}$, and note that $\abs{X} = t \cdot (p^k-1) < s \cdot p^k$.
    By the maximality of the subspaces $V_1,\dots,V_t$, it follows that $V \cap X \neq \emptyset$ for all $V \in S$.
    Therefore, there exists $x^* \in X \setminus \{0\}$ such that $\abs{ \{V \in S : x^* \in V\}} \geq \frac{\abs{S}}{\abs{X}}$.

    For all $V \in S$ such that $x^* \in V$ define $V_{-x^*} \seq V$ to be any $k-1$ dimensional subspace of $V$ such that $x^* \notin V_{-x^*}$. There might be several such subspaces, and we define $V_{-x^*}$ to be one of them arbitrarily.

    Given such $x^*$, consider now the collection $S'$ of subspaces of dimension $k-1$:
    \begin{equation*}
        S' = \{V_{-x^*} : x^* \in V \text{ and } V \in S\}.
    \end{equation*}
    (Note that since $S$ might contain the same subspace more than once, the same applies also to $S'$.)
    
    The collection $S'$ contains at least $\frac{\abs{S}}{\abs{X}} \geq \frac{s^{k+1} \cdot p^{\frac{k^2+k-2}{2}}}{s \cdot p^k} = s^{(k-1)+1} \cdot p^{\frac{(k-1)^2+(k-1)-2}{2}}$ subspaces of dimension $k-1$. Therefore, by the induction hypothesis $S'$ contains a sunflower of size $s$, which induces a sunflower of size $s$ in $S$.
\end{proof}

\section{Chebyshev lemma}\label{sec:chebyshev-lemma}

\cref{lem:dist-of-product-distributions} talks about the distance between two product distributions. Indeed, our target $\cB_{1/3}^{\ox n}$ is a product distribution, but in general hoping for a product structure in polynomial distributions is too much. Therefore, we prove an analog of \cref{lem:dist-of-product-distributions} assuming that our polynomials are \emph{pairwise almost independent}, where our notion of independence is captured by regularity.

\begin{lemma}[Chebyshev lemma]\label{lem:chebyshev-bounded-rank-d}
    Let $\cP=(P_1,\dots,P_n)$ be an $n$-tuple of functions $P_i\colon\F_2^m \to \F_2$ such that $\rank_d(P_i) \leq r$, and suppose $\norm{P_i - \Ber(1/3)}_\tv \ge \delta>0$ for all $i\in [n]$.
        
    Write each $P_i$ as $P_i = \Gamma_i (Q_{i1},\dots,Q_{ir})$, where $\Gamma_i\colon \F_2^r \to \F_2$ is an arbitrary function of $r$ variables\footnote{Strictly speaking, we mean that for each $P_i$ there exists some $r_i \in \{0,\dots, r\}$ such that $P_i = \Gamma_i(Q_{i1}, \dots, Q_{ir_i})$. Since treating every tuple as having length $r$ does not affect the argument, we make that assumption for simplicity.} and $Q_{ij}$'s are polynomials of degree $d$.
    Suppose that all pairs $i\neq j$ satisfy the \emph{pairwise regularity assumption}, namely
    \begin{equation*}
        \rank(Q_{i1},\dots,Q_{ir}; Q_{j1},\dots,Q_{jr}) > c_\KL(d, 2^{-r}\eta).
    \end{equation*}
    Then
    \[
    \norm{\cP-\cB_{1/3}^{\ox n}}_\tv \ge 1 - O\big(\frac{\eta}{\delta^2}\big) - O\big(\frac{1}{\delta^2 n}\big),
    \]
    where the hidden constants in big-O are absolute.
\end{lemma}
\begin{proof}
    By pigeon-hole principle, at least half of the $P_i$'s satisfy either $\Pr_X[P_i(X)=1]\ge \frac{1}{3} + \delta$ or $\Pr_X[P_i(X)=1]\le \frac{1}{3} - \delta$. We assume the former is true for \emph{all} $P_i$ for simplicity. (We will replace $n$ with $n/2$ for the final result. The implied constants will be absorbed by big-O.)
    
    Associate to each $P_i$, the indicator random variable $\indicator{P_i}$ (to emphasize that the arithmetic involving $\indicator{P_i}$ is over $\Z$ rather than $\F_2$). The key step in the proof is the following claim.
    \begin{claim}\label{claim:cov-p-i-p-j}
        For all $i \neq j$ we have $\Cov[\indicator{P_i},\indicator{P_j}] \leq 4\eta$.
    \end{claim}

    Given \cref{claim:cov-p-i-p-j} we now show that the number of polynomials outputting $1$,  $\sum_{i} \indicator{P_i}$, is bounded away above $n/3$ with high probability. Specifically,
    \begin{align*}
        \Pr \left[\textstyle\sum_i \indicator{P_i} < (\frac{1}{3}+\frac{\delta}{2})n \right]\tag{Chebysev's inequality}
        \;&\leq\; \frac{\Var[\textstyle\sum_i \indicator{P_i}]}{(\delta/2)^2n^2} \\
        \; &= \; \frac{\sum_i \Var[\indicator{P_i}] + 2\sum_{i<j} \Cov[\indicator{P_i},\indicator{P_j}]}{(\delta/2)^2n^2}\hspace{3.5cm} \\
        \; &\leq \;  \frac{n + 4\eta n^2}{(\delta/2)^2n^2} \tag{\cref{claim:cov-p-i-p-j}} \\
         \; &= \; O(\frac{1}{\delta^2 n}) + O(\frac{\eta}{\delta^2}).
    \end{align*}

    Now define the event $E=\{x\in\Bits^n : \wt(x) < (\frac{1}{3}+\frac{\delta}{2})n\}$ to be all outcomes with Hamming weight less than $\frac{n}{3}+\frac{\delta n}{2}$. We get,
    \begin{eqnarray*}
        \norm{\cP-\cB_{1/3}^{\ox n}}_\tv \geq  \Pr[\cB_{1/3}^{\ox n} \in E] - \Pr[\cP \in E] \geq 1- \frac{2}{2^{\Omega(\delta^2 n)}} - O(\frac{\eta}{\delta^2}) - O(\frac{1}{\delta^2 n}),
    \end{eqnarray*}
    where in the last inequality we use a standard Chernoff-type concentration of measure for $\cB_{1/3}^{\ox n}$ to compute $\Pr[\cB_{1/3}^{\ox n} \in E]$. This completes the proof of the lemma.
\end{proof}

We now return to the proof of \cref{claim:cov-p-i-p-j}.

\begin{proof}[Proof of \cref{claim:cov-p-i-p-j}]
    Let us use $\cQ_i$ to denote $(Q_{i1},\dots, Q_{ir})$. By \cite{kaufman2008worst}, the pairwise regularity assumption implies that for any nonzero linear combination of polynomials in $(\cQ_i; \cQ_j)$ it holds that
    \begin{equation*}
    \bias(\lambda_{i1}Q_{i1} + \dots + \lambda_{ir}Q_{ir} + \lambda_{j1}Q_{j1} + \dots + \lambda_{jr}Q_{jr}) \le 2^{-r}\eta.
    \end{equation*}    
    Therefore, by Vazirani's XOR \cref{lem:vazirani-xor}, the total variation distance between $(\cQ_i; \cQ_j)$ and the uniform distribution over $\Bits^{2r}$ is upper bounded by $\eta$. Consequently, sampling the tuple $(X;Y)$ uniformly at random from $\Bits^{r}\times \Bits^r$, we get
    \begin{align*}
        \Cov[\indicator{P_i},\indicator{P_j}]
         &=  \Cov[\indicator{\Gamma_i (\cQ_i)},\indicator{\Gamma_j (\cQ_j)}] \\
         &= \E[\indicator{\Gamma_i (\cQ_i)}\cdot \indicator{\Gamma_j (\cQ_j)}] 
         - \E[\indicator{\Gamma_i (\cQ_i)}] \cdot \E[\indicator{\Gamma_j (\cQ_j)}]\\
         &\le \E[\indicator{\Gamma_i(X)}\cdot \indicator{\Gamma_j(Y)}] - \E[\indicator{\Gamma_i(X)}]\cdot \E[\indicator{\Gamma_i(Y)}] + 4\eta \\
         &= 4\eta.
    \end{align*}
    The inequality uses the fact that $\norm{(X;Y)- (\cQ_i; \cQ_j)}_\tv\le \eta$ so $\cQ_i$ and $\cQ_j$ may be replaced by uniform random variables $X$ and $Y$; this affects the value of each expectation by at most $\eta$.
\end{proof}

\subsection{Special variants of Chebyshev lemma}

\begin{corollary}[Chebyshev lemma; bounded $\rank_2$]\label{cor:chebyshev-only-bounded-rank-2}
    Let $\cP=(P_1,\dots,P_n)$ be an $n$-tuple of functions of $\rank_2(P_i)\le r$. Suppose that for all $i \in [n]$ we have $P_i = \Gamma_i (Q_{i1},\dots,Q_{ir})$, where $\Gamma_i \colon \F_2^n \to \F_2$ is an arbitrary function of $r$ variables and $Q_{ij}$'s are quadratic polynomials satisfying
    \begin{equation*}
        \rank_1(Q_{i1},\dots,Q_{ir}; Q_{j1},\dots,Q_{jr}) > c_\KL(2, 2^{-(C+1)r}) = 2(C+1)r
    \end{equation*}
    for all $i \neq j$.
    Furthermore, suppose that for each $i\in [n]$,  $\abs{\Pr[P(X)=1] - 1/3} \ge \delta$.
    Then 
    \[
    \norm{\cP-\cB_{1/3}^{\ox n}}_\tv \ge 1 - O(2^{-Cr}/\delta^2) - O(1/(\delta^2 n)),
    \]
    where the hidden constants in big-O are absolute.
\end{corollary}
\begin{proof}
    Recall that $c_\KL(2,  \delta)=2\log(1/\delta)$ by \cref{thm:structure-of-deg2}. The corollary follows by applying \cref{lem:chebyshev-bounded-rank-d} with $\eta = 2^{-Cr}$.
\end{proof}

Next we extend the lemma above to the setting $\Gamma_i$ depends on a small number of quadratics and linear function.
For a function $P\colon \F_2^m \to \F_2$, we write \defemph{$\rank_{2,1}(P) \le (r,s)$} if there exist quadratic polynomials $Q_1,\dots,Q_r$, linear polynomials $L_1,\dots,L_s$, and a function $\Gamma \colon \Bits^{r+s} \to \Bits$ such that $P = \Gamma(Q_1,\dots,Q_r;L_1,\dots,L_s)$.

\begin{corollary}[Chebyshev lemma; bounded $\rank_{2,1}$]\label{cor:chebyshev-pairwise-high-rank-lin-part}
    Let $\cP=(P_1,\dots,P_n)$ be an $n$-tuple of functions $P_i \colon \F_2^m \to \F_2$ with $\rank_{2,1}(P_i)\le (r, s)$, where $r\ge 1$ and $s\ge 0$. Suppose that $\abs{\Pr[P_i(X)=1]-1/3}\ge \delta$ for all $i\in [n]$.
    Furthermore, suppose that for all $i \in [n]$ we have $P_i = \Gamma_i (Q_{i1},\dots,Q_{ir}; L_{i1}, \dots,L_{is})$, where $\Gamma_i \colon \Bits^{r+s} \to \Bits$ is an arbitrary function of $r+s$ variables, $L_{ik}$'s are linear functions and $Q_{ik}$'s are quadratic polynomials satisfying
    \begin{equation*}
        \rank_1(Q_{i1},\dots,Q_{ir}; Q_{j1},\dots,Q_{jr}) > 2(C + 3)(r+s)
    \end{equation*}
    for all $i \neq j$ and integer $C\ge 1$.
    Then
    \[
    \norm{\cP-\cB_{1/3}^{\ox n}}_\tv \ge 1 - O\big(\frac{2^{2s}}{\delta^2 n^{1/(s+1)}}\big) - O\big(2^{-C(r+s)}/\delta^2\big).
    \]
\end{corollary}

\begin{remark}
We emphasize that for $i \neq j$ the linear functions $L_{i1}(x),\dots,L_{is}(x);L_{j1}(x),\dots,L_{js}(x)$ might have arbitrary linear dependencies.
\end{remark}

\begin{proof}
    The proof strategy is to reduce the problem to the setting of \cref{cor:chebyshev-only-bounded-rank-2}.

    Apply the sunflower lemma \cref{lemma:sunflower} on the collection $V_1,\dots, V_n$ of size $n$, where the $i$\textsuperscript{th} subspace $V_i$ is $\Span(L_{i1},\dots,L_{is})$.

    There exists a subset $S \seq [n]$ of the $V_i$, of size $t \geq (\frac{n}{2^{\frac{s^2+s-2}{2}}})^{\frac{1}{s+1}} \geq 2^{-s/2} \cdot n^{\frac{1}{s+1}}$ that forms a sunflower.

    Let $V_{\text{core}} = \cap_{i\in S} V_i$ be the core of our sunflower, where $\dim(V_{\text{core}}) = \ell \le s$. At this point, by a change of basis we obtain a new $\rank_{2,1} \le (r,s)$ representation of $P_i$'s that belong to the sunflower.
    \[
    \forall i\in S: \qquad P_i = \Gamma_i (Q_{i1},\dots,Q_{ir}; L_{i1}, \dots,L_{is}) = \Gamma'_i(Q_{i1},\dots,Q_{ir}; L'_{i1},\dots, L'_{i, s-\ell}; L^*_1, \dots, L^*_{\ell}),
    \]
    where $(L^*_1, \dots, L^*_{\ell})$ is a basis of the core.

    Let $\cQ_i, \cL'_i, \cL^*$ denote each block of $(Q_{i1},\dots,Q_{ir}; L'_{i1},\dots, L'_{i, s-\ell}; L^*_1, \dots, L^*_{\ell})$ respectively. Note that $\cL'_i$ and $\cL'_j$ are linearly independent for two distinct $i,j\in S$. Therefore, for any two distinct $i,j \in S$ the regularity condition of \cref{cor:chebyshev-only-bounded-rank-2} holds in the sense 
    \[
        \rank(\cQ_i, \cL'_i, \cQ_j, \cL'_j) \geq 2(C + 3)(r+s).
    \]
    It follows from \cref{cor:chebyshev-only-bounded-rank-2} that for any $z \in \F_2^\ell$ the distribution $\cP_{S, z} = \big(\Gamma'_i(\cQ_i, \cL'_i, z) : i \in S\big)$ is $1 - O(\frac{2^{(s-\ell)/2}}{\delta^2n^{1/(s+1)}}) - O(2^{-(C+2)(r+s)}/\delta^2)$ far from $\cB_{1/3}^{\ox n}$.

    As our final step, we apply \cref{lem:dist-tv-common-Q}. Since $\Pr[\cL^* = z] = 2^{-\ell}$, we get

    \[
    \norm{\cP - \cB_{1/3}^{\ox n}}_\tv \ge 1 - O\big(\frac{2^{s/2 + 3\ell/2}}{\delta^2n^{1/(s+1)}}\big) - O\big(2^{-(C+2)(r+s)+2\ell}/\delta^2\big) \ge 1 - O\big(\frac{2^{2s}}{\delta^2 n^{1/(s+1)}}\big) - O\big(2^{-C(r+s)}/\delta^2\big).\qedhere
    \]

\end{proof}

\section{A sunflower pairwise regularization}\label{sec:pairwise-reg} 

The classic sunflower lemma \cite{erdos1960intersection} states that every sufficiently large family $S_1,\dots,S_n$ of $k$-element sets contains a large subfamily $S_{\pi(1)},\dots,S_{\pi(w)}$ whose members all have the same pairwise intersection; that is, $\cap_{i=1}^w S_{\pi(i)} = S_{\pi(j)}\cap S_{\pi(j')}$ for all $j\neq j'$. In this section, we prove a theorem about a collection of polynomials that is in some sense an analogue of the sunflower lemma for polynomials. We show that for every sufficiently large family of degree-$d$ polynomials $P_1,\dots, P_n$ there exists a large subfamily $P_{\pi(1)},\dots, P_{\pi(w)}$ whose polynomials can simultaneously be written in the form $P_{\pi(i)}=\Gamma_i(Q^*_1,\dots, Q^*_c; Q_{i1}, \dots, Q_{ir})$. Here, the sets of polynomials $\cG_i = \{Q^*_1,\dots, Q^*_c; Q_{i1}, \dots, Q_{ir}\}$ that compute $P_{\pi(i)}$ form a sunflower, and additionally enjoy some regularity conditions. In this section we restrict our attention to the combinatorial result itself, and leave its consequences for sampling lower bounds to \cref{sec:deg-d-lower-bound}.

We will view a factor $\cF$ as a set of polynomials instead of a tuple. This enables us to take intersection or set difference between two factors using familiar notation. As a matter of convenience, the empty factor $\emptyset$ is considered to be infinitely regular.

The \defemph{locality} of a collection of factors $\cF_1,\dots,\cF_n$ is the maximum cardinality (i.e. dimension) of a factor.

\begin{definition}
Let $\cF_1,\dots,\cF_n$ be a collection of degree-$d$ factors. Define the \defemph{core} $\cC = \bigcap_{i}\cF_i$. The \defemph{state} of $\cF_i$ is a vector $\vec S_i\in\N^{2d+1}$ defined as follows. Let $K=\max_i\card{\cF_i}$ denote the locality of this collection, and set $\Delta_i = K - \card{\cF_i}$ for each $i\in[n]$. Then the state of $\cF_i$ is the vector $\vec S_i =(\Delta_i, \vec N_0,\vec N_i)$, where $\vec N_0\in \N^d$ is the dimension vector $\cC$, and $\vec N_i\in\N^d$ is the dimension vector of $\cF_i\setminus\cC$ (so the dimension vector $\cF_i$ is $\vec M_i=\vec N_0+\vec N_i$). The first component of each $\vec S_i$, $\Delta_i$, will be commonly denoted as its \defemph{null component}.

The \defemph{weight} of a state $\vec S_i$, denoted by $\wt(\vec S_i)$, is the sum of each component of $\vec S_i$.

The \defemph{state} of the collection $\cF_1,\dots,\cF_n$ is simply the maximum state $\vec \cS = \max_i\vec S_i$ w.r.t. the inverse lexicographic ordering of $\N^{2d+1}$. 

Note that the following equalities are forced by this definition $\wt(\vec\cS)=\wt(\vec S_1)=\dots = \wt(\vec S_n)=K$.
\end{definition}

Recall that a collection of sets $V_1,\dots, V_w$ (possibly repeated) form a \defemph{sunflower} if for $\cC=\bigcap_iV_i$ it holds that $\cC = V_i\cap V_j$ for all $i\neq j$. The set $\cC$ is the \defemph{core} of the sunflower. This definition allows for a collection of empty sets to form a sunflower.

\begin{theorem}[Sunflower pairwise regularization]\label{thm:lossy-pairwise-regularization}
    Let $f\colon \N \to \N$ be a growth function, and let $K$ and $d$ be positive integers. Then there exit $\eps > 0$ and $K'$ such that the following hold.
    
    Let $\cF_1, \dots, \cF_n$ be a collection of degree-$d$ factors with locality $\card{\cF_i}\le K$. Then for a sub-collection $\cF_{\pi(1)},\dots, \cF_{\pi(w)}$ of size $w=n^{\eps}$  there exists a collection of degree-$d$ factors $\cG_1,\dots, \cG_w$ that satisfies the following.
    \begin{enumerate}
        \item (Refinement) For every $i$, $\cG_i$ is a refinement of $\cF_{\pi(i)}$;
        \item (Locality) For every $i$, $\card{\cG_i}\le K'$;
        \item (Pairwise Regularity) For every $i,j$ the factor $\cG_i\cup \cG_j$ is $f$-regular;
        \item (Sunflower) The factors $\cG_1,\dots, \cG_w$ form a sunflower.
    \end{enumerate}
\end{theorem}

\begin{proof}
The proof goes by strong induction on the state $\vec \cS$ of our collection $\cF_1,\dots, \cF_n$. As a base case, any state of the form $\vec \cS = (K, \vec 0, \vec 0)$ corresponds to a collection of empty factors, which vacuously satisfy all the four requirements. Now we assume that $\vec \cS = (M_0, \dots, M_{2d})$ contains some nonzero component $M_i > 0$ for some $i\ge 1$.

Let $\cC = \bigcap_{i}\cF_i$ be the core, and let $\cQ_i = \cF_i \setminus \cC$. We make explicit the elements of each factor, writing $\cC = \{C_1,\dots, C_{k_0}\}$ for the core, and $\cQ_i = \{P_{i1},\dots, P_{ik_i}\}$ for the part outside of the core. Let $K = \max_i \card{\cF_i}$ denote the current locality.

We will assume that the polynomials in each factor are linearly independent, since any linear dependencies can be discarded. This clearly does not affect the final claim.

We have three cases, in each of which after applying a refinement step the state $\vec \cS$ decreases, allowing us to apply the induction hypothesis.

\paragraph{Case I ($\cC$ is not $(f(2K) + 2)$-regular).} We perform the following refinement step on all factors. Identify some nonzero linear combination $R_0=\sum_{i=1}^{k_0}\lambda_iC_i$ whose $\rank_{\ell-1}$ is $r\le f(2K) + 2$, where $\ell = \deg_\cC(R_0)$ is the degree of this linear combination. Replace a maximal degree $C_i$ with $\lambda_i\neq 0$ with $r$ many degree $\ell - 1$ polynomials. Note that this affects the state of all non-empty factors because $\cC$ is the common intersection of all of them. However, all the $\vec N_i$ for $i\ge 1$ remain unaffected, and only $\vec N_0$ will be replaced by a smaller dimension vector.

In this case, \emph{all} dimension vectors of non-empty factors decrease.

\paragraph{Case II (At least $n/2$ many $\cF_i$ are not $(f(2K)+1)$-regular).}
Let $I=\{i\in[n] : \rank(\cF_i) \le f(2K)+1\}$. Remove all $\cF_i$ with $i\in [n]\setminus I$ from the collection, and do the following for all $\cF_i$ with $i\in I$.

Identify some nonzero linear combination $R_i=\sum_{j=1}^{k_0}\lambda_j C_j + \sum_{j=1}^{k_i}\alpha_j P_{ij}\in \Span(\cF_i)$ whose $\rank_{\ell-1}$ is $r\le f(2K) + 1$, where $\ell = \deg_{\cF_i}(R_i)$ is the degree of this linear combination. Note that by case I, this linear combination cannot be supported entirely on $\cC$; hence there exists some $P_{ij}$ with $\alpha_j\neq 0$. Furthermore, there must exist such a $P_{ij}$ that has maximal degree $\ell$ in this linear combination. Replace a $P_{ij}$ with $\alpha_j\neq 0$ and maximal degree with $r$ many degree $\ell-1$ many polynomials.

In this case, the states of all surviving factors, namely those in $I$, decrease.

\paragraph{Case III ($\cF_i\cup \cF_j$ is not $f(2K)$-regular for some $i, j$).}
Because we apply case II prior to this case, there can be at most $n/2$ many $\cF_i$ that are not $f(2K) + 1$ regular. Remove all such $\cF_i$'s from our collection, keeping at least $n'=n/2$ many $f(2K)+1$ regular $\cF_i$s.

Then, build a graph $G$ that captures \emph{pairwise irregularities}. More formally, $G$, defined on the vertex set $[n']$, has an edge $e=uv$ iff $\cF_u\cup\cF_v$ is \emph{not} $f(2K)$-regular. Now we label each edge $e=uv$ by a pair of polynomials $(R^e_u, R^e_v)$,  where $R^{e}_u\in\Span(\cF_u)$ and $R^{e}_v\in \Span(\cF_v)$ such that $\rank_{\ell-1}$ of $R^{e}_u +  R^{e}_v$ is at most $f(2K)$, where $\ell=\max\{\deg_{\cF_u}(R^e_u), \deg_{\cF_v}(R^e_v)\}$ is the degree of the linear combination\footnote{Recall that the degree of a linear combination $\sum_i\lambda_iP_i$ is the maximum degree among the polynomials that participate in it, i.e., $\max_i\deg(\lambda_iP_i)$.} that defines $R^e_u+ R^e_v$. Note that neither of these polynomials are constant, and neither belong to $\Span(\cC)$.

If $G$ contains an independent set of size $\sqrt{n'}$, we declare the corresponding sub-collection of $\cF_i$’s to be pairwise $f(2K)$-regular, and we set all the remaining $\cF_i$ (those outside the independent set) to $\emptyset$. This is where the induction terminates.

Otherwise, $G$ has a vertex of degree at least $\sqrt{n'}$. Up to relabeling, we assume this vertex is $1$, and is adjacent to each of $2,3,\dots,\sqrt{n'}+1$. By the pigeonhole principle, in at least $t=\sqrt{n'}/2^K$ many edges $e=1v$, the same polynomial $R^{e}_1\eqdef R_1$ appears. Upon relabeling once again, assume that the first $t$ neighbors $2, \dots, t+1$ have that property. Now for each $v=2,\dots,t+1$ the edge $e=1v$ is labeled by a pair $(R_1,R_v)$. Recalling how we defined the labels of edges, we write
\begin{align*}
R_i &= \underbracket{\sum_{j=1}^{k_0}\lambda_{ij}C_j}_{\eqdef A_i} + \underbracket{\sum_{j=1}^{k_i}\alpha_{ij}P_{ij}}_{\eqdef B_i} \qquad \text{for } i = 1,2,\dots, t+1.
\end{align*}

Let $\ell_i=\deg_{\cF_i}(R_i)$ denote the degree of the linear combination that defines $R_i$. For each $i=1, \dots, t+1$, we have that $A_1 + B_1 + A_i + B_i$ can be written as a bounded combination of the form,
\[
A_1 + B_1 + A_i + B_i = \Gamma_i(Q_{i1},\dots, Q_{ir_i}),
\]
for some $r_i \le f(2K)$, and $Q_{i1},\dots, Q_{ir_i}$ degree $\max\{\ell_1,\ell_i\}-1$ polynomials. Both $B_1$ and $B_i$ must contain a degree-$\max\{\ell_1,\ell_i\}$ polynomial: indeed, if, for example, $B_i$ does not contain such a maximal-degree polynomial, then $A_1+B_1+A_i \in \Span(\cF_1)$ has rank more than $f(2K)+1$, and therefore, $A_1+B_1+A_i+B_i$ has rank more than $f(2K)$, a contradiction. Therefore, $\ell_1=\ell_i\eqdef \ell$.

For each $i=1,\dots, t+1$, let $H_i = P_{ij}$ for some $j$ such that $\alpha_{ij}\neq 0$ and $P_{ij}$ has maximal degree ($=\ell$). We will remove $H_i$ from $\cF_i$ and then add polynomials $A_1 + B_1, Q_{i1},\dots, Q_{ir_i}$. Therefore, deleting all factors $\cF_i$ with $i\notin[t+1]$ from our collection, we refine the $t+1$ factors as such,
\[
\cF_i \longmapsto \cF_i' = (\cF_i \setminus \{H_i\}) \cup \{A_1 + B_1, Q_{i1},\dots, Q_{ir_i}\}  \qquad \text{for } i=1,2\dots, t+1.
\]

\paragraph{Combining cases I, II, and III.} After performing one of the refinement stepsto the collection $\cF_1,\dots, \cF_n$ with state $\vec \cS$, we have obtained a refinement of a sub-collection collection $\cF'_1,\dots, \cF'_t$ with the new state $\vec \cS'$, where $t\ge \sqrt{n}/2^{K+1}$.

We invoke the induction hypothesis to obtain $\cG_1,\dots,\cG_w$ with the state $\vec \cT$, which satisfies $\wt(\vec \cT) \le \psi(\vec \cS')$, where $\psi=\psi_{2d+1,g}$ for $g(r)=f(2r)+2$. Let $j$ be such that $\vec S'_j=\max_i \vec S'_i=\vec\cS'$. Then, 
\[
\wt(\vec \cT) \enspace \stackrel{\text{IH}}{\le}\enspace \psi(\vec \cS') \enspace =\enspace \psi(\vec S'_j)\enspace \stackrel{\text{\ref{prop:properties-of-rapid-function}}}{\le}\enspace \psi(\vec S_j)\enspace \stackrel{\text{\ref{prop:properties-of-rapid-function}}}{\le} \enspace \psi(\vec \cS)\enspace \stackrel{\text{\ref{prop:properties-of-rapid-function}}}{\le} \enspace \psi^*(K).
\]

Because the induction terminates in at most $\psi^*(K)$ steps, $w \ge n^{\exp(-\psi^*(K))}2^{-2(\psi^*(K)+1)}\ge n^\eps$ for large enough $n$ and some appropriate $\eps>0$.
\end{proof}

\section{One polynomial's distance from Ber(1/3)}\label{sec:lower-bound-for-a-single-poly}

In this section we show that a degree-$d$ polynomial $P \colon \F_2^m \to \F_2$ satisfies
\[
\abs{\Pr_X[P(X)=1]-1/3} \ge \delta_d,
\]
where $\delta_d > 0$ is a constant depending only on $d$. A more general claim is true: for any non-dyadic $\rho\in (0, 1)$ and degree-$d$ polynomial $P$ it holds that $\abs{\Pr_X[P(X)=1]- \rho} \ge \Omega_{d,\rho}(1)$. We show this in \cref{sec:non-constructive-lower-bound-for-a-single-poly}. A closely related problem is the weight distribution of binary Reed-Muller codes, which asks how many polynomials of degree $d$ satisfy $\Pr_X[P(X) = 1] \le \eta$ for some parameter $\eta$. This theorem shows there are no polynomials with $\Pr_X[P(X)=1]$ close to $1/3$ exist.

A similar bound for a $d$-local function (trivially) holds, which has proved useful in showing lower bounds for distributions defined by local functions and bounded depth decision forests \cite{kane2024locality,viola2023new} against $\cB_{1/3}^{\ox n}$ or the Hamming slice $H_{n/3}=\{x\in\Bits^n : \sum_i x_i = n/3\}$.

We divide the results of this section into two parts. First we deal with small degrees $1,2,$ and $3$. Then we move to a general constant $d$. These results rely on the structure theorems discussed in \cref{sec:structure-theorems}.

\subsection{One polynomial's gap: degree $1,2,3$}

A degree-$1$ polynomial $L$ satisfies $\Pr_X[L(X)=1]\in \{0, \frac{1}{2}, 1\}$. Therefore, $\abs{\Pr_X[L(X)=1]-\frac{1}{3}} \ge \delta_1$ with $\delta_1=1/6$.

The structure theorem \cref{thm:structure-of-deg2} implies that a quadratic polynomial $Q$ satisfies $\bias(Q)\in \{0\}\cup \{\pm2^{-k}:k\in \N\}$. Therefore $\delta_2=1/24$.

Below we show that all polynomial of degree 3 we have $\abs{\Pr_X[P(X) = 1]-\frac{1}{3}}>\delta_3$ for some absolute constant $\delta_3>0$. The key claim is the following statement.

\begin{claim}\label{claim:rank-2-dyadic}
    Let $P \colon \F_2^m \to \F_2$ be a function of $\rank_2(P) \leq r$. That is, it can be written as $P(x) = \Gamma(Q_1(x),\dots,Q_r(x))$ for some boolean $\Gamma\colon \Bits^r \to \Bits$ and some quadratic polynomials $Q_1,\dots,Q_r \colon \Bits^m \to \Bits$.
    For every integer $t>0$, there exists some $s \leq (r+t) \cdot 2^r + r$ such that $\abs{\E_x[(-1)^{P(x)}] - \frac{A}{2^s}} \geq \frac{1}{2^t} \cdot \frac{1}{2^s}$ for all integers $A \in \Z$.
\end{claim}

\begin{proof}
    We start by writing the Fourier expansion of $\Gamma$.
    \begin{equation*}
        \Gamma(y_1,\dots,y_r) = \sum_{S \seq [r]} \hat\Gamma(S) \cdot (-1)^{\sum_{i \in S} y_i},
    \end{equation*}
    where $\hat\Gamma(S) = \frac{a_S}{2^r}$ for some integer $-2^r \leq a_s  \leq 2^r$.
    Then, for $Q_S = \sum_{i \in S} Q_i$ we have
    \begin{equation*}
        \E_x[(-1)^{\Gamma(Q_1(x),\dots,Q_r(x))}] =  \sum_{S \seq [r]} \hat\Gamma(S) \cdot \E_x[(-1)^{Q_s(x)}].
    \end{equation*}
    Denote $\eps_S = \E_x[(-1)^{Q_s(x)}]$ for all $S \seq [r]$.
    Then, using \cref{thm:structure-of-deg2} we have $\eps_S \in \{\pm 2^{-i} : i \in \N\}$.
    We claim that there is some $0 \leq i \leq 2^r$ such that the interval $(2^{-(r+t)(i+1)}, 2^{-(r+t)i}]$ does not contain any of $\abs{\eps_S}$. Indeed, consider the collection of disjoint intervals $\{(2^{-(r+t) \cdot (i+1)}, 2^{-(r+t) \cdot i}] : i=0,\dots,2^r\}$. Since there are $2^r+1$ such intervals, and only $2^r$ different values of $\abs{\eps_S}$, one of the intervals does not contain any $\abs{\eps_S}$.

    Next then write $\E[(-1)^{P(x)}]$ as follows.
    \begin{eqnarray*}
        \E_x[(-1)^{P(x)}]
        & = & \E_x[(-1)^{\Gamma(Q_1(x),\dots,Q_r(x))}] = \sum_{S \seq [r]} \hat\Gamma(S) \cdot \eps_S \\
        & = & \sum_{S \seq [r] : \eps_S > 2^{-(r+t)i}} \hat\Gamma(S) \cdot \eps_S 
        +  \sum_{S \seq [r] : \eps_S < 2^{-(r+t)(i+1)}} \hat\Gamma(S) \cdot \eps_S \\
        & = & \frac{A}{2^{(r+t)i}} + \textit{err},
    \end{eqnarray*}
    where $A \in \Z$ is some integer, and $\abs{\textit{err}} = \sum_{S \seq [r] : \eps_S \leq 2^{-(r+t)(i+1)}} \hat\Gamma(S) \cdot \eps_S \leq 2^r \cdot 2^{-(r+t)(i+1)} = \frac{1}{2^t}2^{-(r+t)i}$.

    Letting $s=(r+t)i + r \leq (r+t) \cdot 2^r+r$ we get
    \begin{equation*}
        \abs{\E_x[(-1)^{P(x)}] - \frac{A}{2^{s}}} \leq \frac{1}{2^t} \cdot \frac{1}{2^{s}},
    \end{equation*}
    as required.
\end{proof}

We will also need the following proposition.
\begin{proposition}\label{prop:gap-of-one-third-from-dyadics}
    We have $\abs{\frac{1}{3} - \frac{A}{2^s}} \geq \frac{1}{3} 2^{-s}$ for any integer $s \geq 0$ and any integer $A \in \Z$.
\end{proposition}
\begin{proof}
    Note that it suffices to show that $\abs{\frac{2^s}{3} - A} \geq \frac{1}{3}$, which obviously holds since $2^s$ is not divisible by 3 and $A$ is an integer.
\end{proof}

We are now ready to prove our lemma for polynomials of degree 3.
\begin{lemma}\label{lemma:deg-3-dyadic}
    \begin{tabular}[t]{@{}r@{\ }p{0.85\linewidth}@{}}
        1. & There is a constant $\delta_{2,r} > 0$ such that any function $P \colon \F_2^m \to \Bits$ of $\rank_2(P) \leq r$ satisfies $\abs{\Pr_x[P(x) = 1] - 1/3}  > \delta_{2,r}$. \\
        2. & There exists a constant $\delta_3>0$ such that for any  polynomial $Q \colon \F_2^m \to \Bits$ of degree $3$ satisfies $\abs{\Pr_x[Q(x)=1]-\frac{1}{3}} \ge \delta_3$. 
    \end{tabular}
\end{lemma}
\begin{proof}
    For the first item, by \cref{claim:rank-2-dyadic}, there exists some integer $A \in \Z$ and $s$ that depends only on $r$, such that
    $\abs{\E_x[P(x)] - \frac{A}{2^{s}}} \leq \frac{1}{4} \cdot 2^{-s}$.
    On the other hand, using the fact that $\abs{\frac{A}{2^s} - \frac{1}{3}} \geq \frac{1}{3} \cdot 2^{-s}$ for all integers $s \geq 0$ and $A \in \Z$, it follows that
    \begin{equation*}
        \abs{\Pr_x[P(x) = 1] - \frac{1}{3}} \geq \frac{1}{12} \cdot 2^{-s}.
    \end{equation*}
    Since $s \leq (r+2) \cdot 2^r + r$, it follows that
    $\abs{\Pr_x[P(x) = 1] - \frac{1}{3}} \geq \delta_{2,r}$ for $\delta_{2,r} = \frac{1}{12} \cdot 2^{-((r+2) \cdot 2^r + r)}$.

    \medskip

    For the second item, if $\abs{\Pr[Q(x) = 1]-\frac{1}{3}}\ge 0.06$ then we are done. Otherwise, we have $\abs{\Pr[Q(x) = 1]-\frac{1}{2}} \geq 0.1$. By \cref{thm:haramaty-shpilka} there exist some absolute constant $r$ such that $Q$ can be written as $Q(X) = \Gamma(Q_1(X),\dots, Q_r(X))$,
    where $Q_1,\dots,Q_r$ are  quadratic polynomials, and $\Gamma \colon \Bits^r \to \Bits$ is some function on $r$ bits.%
    \footnote{In fact, \cref{thm:haramaty-shpilka} gives a stronger structure theorem. Specifically, there exist constants $c$ and $r$ such that
    \[
    P=\sum_{i=1}^{r-c}L_iQ_i + \Gamma(L_{r-c+1},\dots, L_{r}),
    \]
    where $L_i$'s are linear function and $Q_i$s are quadratic polynomials, and $\Gamma$ is some boolean function.}
    Therefore, by the first item $\delta_3 \geq \delta_{2,r}$ for som absolute constant $r$, as required.
\end{proof}

\subsection{One polynomial's gap: degree $d$}
Here we prove that a degree-$d$ polynomial satisfies $\abs{\Pr_X[P(X) = 1] - 1/3} \ge \delta_d$. In fact, we prove a more general statement.

\begin{theorem}\label{thm:exp-low-deg-far-from-third}
    Let $P \colon \F_2^m \to \F_2$ be a function with $\rank_d(P) = r$. Then,
    \[
    \abs{\Pr_X[P(X)=1] - \frac{1}{3}} \ge \delta_{d, r},
    \]
    where $\delta_{d, r}>0$ is a constant that depends only on $d$ and $r$, which can be taken to be $\exp(-\psi^*_{d, f}(r))/12$ for $f(k)=c_\KL(d, \frac{1}{4}2^{-3k/2})$.
\end{theorem}
\begin{remark}
    Specializing to $r=1$, we obtain that the distance from $\Ber(1/3)$ for degree-$d$ polynomials is $\delta_d=\delta_{d,1}$.
\end{remark}
\begin{proof}
    Write $P$ as a function of $r$ degree-$d$ polynomials: $P = \Gamma(R_1,\dots, R_r)$. Consider the degree $d$ factor $\cF=(R_1,\dots, R_d)$ with the dimension vector $(M_1, \dots, M_d)$ where $\sum_iM_i=r$. By \cref{thm:classic-regularization}, there exists a factor $\cQ=(Q_1,\dots, Q_K)$ of dimension $K\le \psi_{d,f}^*(r)$ that is $f$-regular for $f(r)=c_\KL(d, \frac{1}{4}2^{-3r/2})$.

    Every nonzero linear combination $\sum_i\lambda_i Q_i$ in $\cQ$ satisfies
    \[
    \bias(\lambda_1Q_1 + \dots + \lambda_KQ_K) \le \frac{1}{4}2^{-3K/2}.
    \]
    By \cref{lem:vazirani-xor}, we have that $\norm{\cQ - \cU^{\ox K}}_\tv \le \frac{1}{4}2^{-K}$.
    Since $P = \Gamma(\cQ)$ for some $\Gamma\colon \F_2^K\to\F_2$, it follows that
    \[
    \Pr[P=1] = \Pr[\Gamma(\cQ)=1] = \Pr[\Gamma(\cU^{\ox K})=1]\pm \frac{1/4}{2^K} = \frac{a}{2^K} \pm \frac{1/4}{2^K} \qquad \text{ for some } a\in \Z.
    \]
    On the other hand, $\abs{\frac{a}{2^K} - \frac{1}{3}}\ge \frac{1}{3}2^{-K}$ by \cref{prop:gap-of-one-third-from-dyadics}, giving the final result
    \[
    \abs{\bias(P)-\frac{1}{3}} \ge \frac{1/12}{2^K},
    \]
    as required.
\end{proof}

\section{Lower bound against linear distributions}\label{sec:deg-1-lower-bound}

We start the technical part with the following easy proof of \cref{main-thm:linear}. We restate it for convenience.

\begin{theorem}[\cref{main-thm:linear} restated]
    Let $\cL = (L_1,\dots,L_n)$ be a degree-$1$ distribution.
    Then
    \[
    \norm{\cL - \cB_{1/3}^{\ox n}}_\tv \geq 1 - \exp(-\Omega(n)),
    \]
    where $\Omega()$ hides some absolute positive constant.
\end{theorem}

\begin{proof}
Let $\cL = (L_1,\dots,L_n)$ be a distribution where each $L_i$ has $\deg(L_i)\le 1$, and let $D$ be the dimension of the span of $L_i$'s. That is, $D$ is the maximal number of linearly independent $L_i$'s. 

Let $c \in (0,1)$ be a parameter, and consider the following two cases.

    \paragraph{Case 1 ($D \ge cn$).}  Assume for simplicity that $L_1,\dots,L_D$ are linearly independent. In particular, the first $D$ coordinates of $\cL$ are statistically independent, and individually their marginal distribution is either a constant or $\cB_{1/2}$.
    In particular, the coordinate-wise distance is $\norm{L_i - \cB_{1/3}}_\tv \ge 1/6$, and by \cref{lem:dist-of-product-distributions} we get $\norm{\cL-\cB_{1/3}^{\ox n}}_\tv \geq 1 - 2 e^{-\frac{(1/6)^2 cn}{12}} \ge 1 - 2^{-\frac{(1/6)^2 cn}{12}}$, where last inequality holds for $n$ that is larger than some absolute constant.
    
    \paragraph{Case 2 ($D < cn$).} Assume for simplicity that $L_1,\dots,L_D$ are linearly independent. Hence, for any conditioning of the inputs so that $(L_1,\dots,L_D) = (a_1,\dots,a_D)$ the distribution $\cL$ becomes \emph{constant}, and hence $\norm{\cL\rest{(L_1,\dots,L_D) = (a_1,\dots,a_D)}-\cB_{1/3}^{\ox n}}_\tv \geq 1 - (2/3)^{n}$. By taking the union bound \cref{lem:dist-of-convex-comb} over all $2^D$ fixings of $L_1,\dots,L_D$, we get $\norm{\cL-\cB_{1/3}^{\ox n}}_\tv \geq 1 - 2^{cn} \cdot (2/3)^{n}$.

A straightforward calculation shows that $c\approx 0.58$ balances the error terms of the above cases $2^{-\frac{(1/6)^2 cn}{12}} = 2^{cn} \cdot (2/3)^{n} \le 2^{-\Omega(n)}$ giving the desired bound.
\end{proof}

\section{Lower bound against bounded linear rank distributions}\label{sec:bounded-rank-1-lower-bound}

In this section, we prove that the distance between $\cB_{1/3}^{\ox n}$ and any distribution $\cP=(P_1,\dots,P_n)$ where each $P_i\colon \F_2^m \to \F_2$ is a function of at most $r$ linear functions is at least $1 - \exp(-2^{-O(r^2)}n))$.

Closely related lower bounds have appeared in the literature. Viola \cite{viola2023new} shows that a depth-$r$ decision forest cannot sample the Hamming slice $H_{n/3}$ to within total variation distance $1 - \exp(-n^{1/\exp(O(r))})$. The work of Kane, Ostuni, and Wu \cite{kane2024locality} shows that $r$-local maps cannot sample $H_{n/3}$ or $\cB_{1/3}^{\ox n}$ any closer than $1-\exp(-n\cdot 2^{-O(r^2)})$.

We show two lower bounds with different proof ideas. The first obtains an exponentially small error which is used in \cref{sec:degree-2-lower-bound} to prove a lower bound for quadratic distributions. The second approach uses sunflowers and gives a weaker error bound; we nevertheless include it because it morally extends to higher-degree distributions.

\subsection{A lower bound with exponentially small error}

Consider a function $P\colon \F_2^m\to\F_2$ with linear rank $r$.\footnote{We also work with restrictions of $P$ to affine subspaces $V+h$ of $\F_2^m$. As discussed in \cref{sec:notation-and-definitions}, affine restrictions preserve degree, and the same is true for $\rank_1$.} Then there exist $L_1,\dots, L_r$ linearly independent linear functions, and some function $\Gamma$ such that $P=\Gamma(L_1,\dots, L_r)$. We use the notation \defemph{$\lin(P)$} to denote the $r$-dimensional linear space $\Span\{L_1,\dots, L_r\}$ on which $P$ depends. The linear space $\lin(P)$ is independent of the choice of $L_1,\dots, L_r, \Gamma$. Using some basic Fourier analysis, $\lin(P)$ can be equivalently characterized as the span of all linear functions $L$ with which $P$ has a nonzero correlation $\E_X[(-1)^{P(X)-L(X)}]\neq 0$. We now set out to prove our sampling lower bound for maps of bounded linear rank.

\begin{theorem}[Bounded linear rank]\label{thm:bounded-rank-1}
    Let $\cP=(P_1,\dots,P_n)$ be a distribution where each $P_i\colon \F_2^m\to\F_2$ satisfies $\rank_1(P_i) \le r$ for some $r \geq 1$. Furthermore, suppose that over all affine subspaces $V+h$ of $\F_2^m$ the coordinate-wise distances satisfy \mbox{$\norm{P_i\rest{V+h}-\Ber(1/3)}_\tv\ge \delta$}. Then for some absolute constant $c > 0$,
    \begin{equation*}
        \norm{\cP-\cB_{1/3}^{\ox n}}_\tv\ge 1 - \exp\Big(-\frac{c^r\delta^{2r}}{r!}\cdot n + r\Big).
    \end{equation*}
\end{theorem}

\begin{proof}
Let $\cP=(P_1,\dots,P_n)$ be a distribution where each $P_i\colon V+h\to\F_2$ is defined over some affine subspace $V+h$ of $\F_2^m$ and satisfies $\rank_1(P_i)\le r$. We show that for some absolute constant $c>0$ the function $\tau\colon \N\to(0,1]$ given by $\tau(r)=c^r\delta^{2r}/{r!}$ satisfies the following.
\[
\norm{\cP-\cB_{1/3}^{\ox n}}_\tv \ge 1 - \exp(-\tau(r)\cdot n+r).\tag{$\star$}
\]

The proof is by induction on $r$.

For the base case of induction consider $r=0$. In this case each $P_i$ is constant. Hence the support of $\cP$ is singleton. Consequently,
\[
\norm{\cP - \cB_{1/3}^{\ox n}}_\tv \ge 1- (2/3)^n \stackrel{\text{want}}{\ge} 1 - \exp(-\tau(0)\cdot n).
\]
Thus we require
\[
\tau(0) \le \log(3/2).  \tag{Condition I}
\]
Fix $r\ge1$ and suppose, by the induction hypothesis, the claim $(\star)$ holds for distributions of rank up to $r-1$. We now proceed to the inductive step.

Construct a bipartite graph $G$ whose left vertex set $\mathcal{L}$ is the set of all linear functions from $V$ to $\F_2$ and whose right vertex set $[n]$ corresponds the functions $P_1,\dots, P_n$.  For each $i\in[n]$ in the graph $G$ add an edge between $i$ and every $L\in\lin(P_i)$ on the left. We denote the neighborhood of a subset $S$ of vertices in $G$ by $N_G(S)$.

Let $W\le \mathcal{L}$ denote a subspace of maximum dimension that satisfies
\begin{equation}\label{eq:expansion}
\abs{N_G(W)} \ge \frac{4}{\tau(r-1)} \cdot \dim(W).
\end{equation}

We consider two cases separately.

\textbf{Case 1 ($\ \abs{N(W)} \ge n/2\ $).} In this case, we restrict our attention to the corresponding polynomials in $I \defeq N_G(W)$. Let $n'$ be the size of $I$ which we know to be at least $n/2$. In contrast, \cref{eq:expansion} implies that $\dim(W) \le \frac{\tau(r-1)}{4}\cdot n$, so the dimension of $W$ is relatively small compared to $n'$. This allows for conditioning on the values that $W$ takes.

Consider some linear function $\alpha \colon  W\to \F_2$ which we call an \emph{assignment}. Because every $i\in I$ has some neighbor in $W$, the distribution $\cP_I\rest\alpha = (P_i\rest\alpha : i\in I)$ is a distribution with linear rank $r-1$. Thus for any assignment $\alpha$, induction hypothesis provides the following.
\[
\norm{\cP_I\rest\alpha - \cB_{1/3}^{\ox n'}}_\tv \ge 1 - \exp\Big(-\tau(r-1)\cdot n' + r-1\Big) \ge 1- \exp\Big(-\frac{\tau(r-1)}{2}\cdot n + r-1\Big).
\]

Note that $\cP_I = \E_\alpha[\cP_I\rest\alpha]$. Applying \cref{lem:dist-of-convex-comb} we get that
\begin{align*}
\norm{\E_\alpha[\cP_I\rest\alpha] - \cB_{1/3}^{\ox n'}}_\tv &\ge  1 - \exp\left(-\left(\frac{\tau(r-1)}{2} - \frac{\tau(r-1)}{4}\right)n + r-1+1\right) \\
&\stackrel{\text{want}}{\ge} 1 - \exp\left(-\tau(r)\cdot n + r\right).
\end{align*}

The last inequality holds if
\[
\tau(r) \le \frac{\tau(r-1)}{4}. \tag{Condition II}
\]

\textbf{Case 2 ($\ \abs{N(W)}  < n/2\ $).}   In this case, consider the subgraph $H$ obtained by removing the vertices in $N_G(W)$. Thus the right vertex set in $H$ is $J \defeq [n]\setminus N_G(W)$. Observe that for every $W' \le \mathcal{L}$, the following holds in graph $H$.
\begin{equation}\label{eq:related-to-a-few}
\abs{N_H(W')} < \frac{4}{\tau(r-1)}\cdot \dim (W').
\end{equation}
Otherwise, $W''=\Span(W'\cup W)$ would be a subspace satisfying \cref{eq:expansion} with $\dim(W'')>\dim(W)$, contradicting the maximality of $W$. A useful consequence of \cref{eq:related-to-a-few} is that for each $P_i$ we have $\abs{N_H(\lin(P_i))} < \frac{4r}{\tau(r-1)}$. Construct a set $I\seq[n]$ according to the following process.
\begin{algorithmic}[1]
\State $I \gets \emptyset$
\While{$J \neq \emptyset$}
    \State $i \gets \min J$ \Comment{{\footnotesize In fact, any index from $J$ would do. }}
    \State $J \gets J \setminus \{i \}$
    \State $I \gets I \cup \{i\}$
    \State $V \gets \Span(\bigcup_{i\in I}\lin(P_i))$
    \State $J \gets J \setminus N_H(V)$ \Comment{{\footnotesize An index $j\in J$ survives if $\lin(P_j)$ is linearly independent from $V$.}}
\EndWhile
\end{algorithmic}

We now argue that by the end of this process, $n'\defeq \card{I}$ is greater than $\frac{\tau(r-1)}{8r}n$. Otherwise by \cref{eq:related-to-a-few} $\abs{I\cup N_H(\Span(\bigcup_{i\in I}\lin(P_i)))} \le n/2$ contradicting the termination as $J$ starts with a size bigger that $n/2$.

The polynomials $\cP_I = (P_i : i \in I)$ form a product distribution since by construction the non-constant $P_i$'s depend on linearly independent $\lin(P_i)$'s. We apply \cref{lem:dist-of-product-distributions} and we get 
\[
\norm{\cP- \cB_{1/3}^{\ox n}}_\tv \ge \norm{\cP_I- \cB_{1/3}^{\ox n'}}_\tv \ge 1 - \exp\left(-\Omega\big(\delta^2\frac{\tau(r-1)}{8r}n\big)+1\right) \stackrel{\text{want}}{\ge} 1 - \exp\left(-\tau(r)\cdot n+r\right).
\]
Therefore, we require that
\[
\tau(r)\le \Omega\big(\delta^2\frac{\tau(r-1)}{8r}\big). \tag{Condition III}
\]
Given conditions I, II, and III, solving for $\tau(r)$ we get $\tau(r)=c^r\delta^{2r}/{r!}$ for some constant $c>0$, as required.
\end{proof}

\begin{remark}
    \cref{thm:bounded-rank-1} treats the coordinate-wise distance $\delta$ as a free parameter. There are, however, only two cases of interest to us. First is the most general setting, where one poses no further assumptions on the $P_i$ beyond having bounded $\rank_1$. In this case, one can take $\delta = \Omega(2^{-r})$. Second, when $P_i$ are additionally quadratic (or degree-$d$) polynomials in which case $\delta$ can be taken to be an absolute constant (that depends only on $d$).
\end{remark}

\begin{corollary}\label{cor:bounded-rank-1}
    Let $\cP=(P_1,\dots,P_n)$ be a distribution where each $P_i$ satisfies $\rank_1(P_i) \le r$. Then $\norm{\cP-\cB_{1/3}^{\ox n}}_\tv\ge 1 - \exp(-c2^{-Cr^2} n)$ for some absolute constants $C, c>0$.
\end{corollary}
\begin{proof}
    Write $P_i=\Gamma(L_{i1},\dots,L_{ir})$ where the $L_{ij}$'s are linearly independent linear functions. Then $\Pr_X[P_i(X) = 1]=a/2^r$ for some $a\in\{0,\dots,2^r\}$. Therefore, $\abs{a/2^r - 1/3} \ge 2^{-r-2}$. The conclusion now follows by applying \cref{thm:bounded-rank-1} with $\delta=2^{-r-2}$.
\end{proof}

\begin{corollary}\label{cor:quadratics-of-bounded-rank-1}
    Let $\cQ=(Q_1,\dots, Q_n)$ be a quadratic distribution where each $Q_i$ satisfies $\rank_1(Q_i) \le r$. Then $\norm{\cQ - \cB_{1/3}^{\ox n}}_\tv\ge 1 - \exp(-c2^{-Cr\log(r)} n)$ for some absolute constants $C,c > 0$.
\end{corollary}
\begin{proof}
    Polynomials are closed under affine restrictions. Therefore, \cref{thm:exp-low-deg-far-from-third} asserts that $Q_i\rest{V+h}$ has a distance of $\delta = \Omega(1)$ with $\cB_{1/3}$. The corollary follows promptly.
\end{proof}

\subsection{A weaker bound using sunflowers for subspaces}

\cref{thm:bounded-rank-1} has our best lower bound $1-\exp(-\Omega_r(n))$ for bounded $\rank_1$ distributions. We prove a similar lower bound, albeit much weaker, using sunflowers for subspaces. The reason we include this weaker bound is twofold. First, its proof is more straightforward, and in particular, is not inductive. Second, the high-level idea of this version will be morally generalized to higher degrees while it is not clear how one can generalize the previous proof to higher degrees. One can work out a similar bound for $d$-local distributions using the classic sunflower lemma \cite{erdos1960intersection}.

To prove our result, we derive a similar variant of \cref{lem:chebyshev-bounded-rank-d} (Chebyshev lemma) for bounded linear rank functions.
\begin{corollary}[Chebyshev Lemma; bounded $\rank_1$]\label{cor:chebyshev-lemma-bounded-rank1}
    Let $\cP=(P_1,\dots,P_n)$ be an $n$-tuple where for all $i \in [n]$ we have $P_i = \Gamma_i (L_{i1},\dots,L_{ir})$, where $\Gamma_i$ is an arbitrary function of $r$ variables and $L_{ij}$'s are linearly independent linear polynomials. Suppose further that, 
    \begin{equation*}
        (\cL_i; \cL_j) = (L_{i1},\dots,L_{ir}; L_{j1},\dots,L_{jr})
    \end{equation*}
    is free from linear dependencies for all pairs $i \neq j$; i.e., $\rank(\cL_i; \cL_j) > 0$.
    Then $\norm{\cP-\cB_{1/3}^{\ox n}}_\tv \ge 1 - O(2^r/n)$.
\end{corollary}
\begin{proof}[Proof Sketch.]
    The proof of this corollary is essentially the same as \cref{lem:chebyshev-bounded-rank-d}, with two differences: $\delta$, the point-wise distance is now $O(2^{-r})$, and the in the \cref{claim:cov-p-i-p-j} we have $\Cov[\indicator{P_i}, \indicator{P_j}]=0$ for $i\neq j$. The result follows.
\end{proof}

\begin{theorem}[Bounded linear rank; using sunflowers]\label{thm:bounded-linear-rank-1-sunflower}
    Let $\cP = (P_1,\dots, P_n)$ be a distribution where each $P_i$ satisfies $\rank_1(P_i)\le r$, where $r\le c\sqrt{\log n}$ for some constant $c>0$. Then,
    \[
        \norm{\cP-\cB_{1/3}^{\ox n}}_\tv\ge 1 - O\Big(\frac{2^{r/2}}{n^{1/(r+1)}}\Big).
    \]
\end{theorem}

\begin{proof}
Define for $P_i$ the subspace $V_i$ of dimension ${\le}r$ it depends on. By the pigeonhole principle, a subset $J\seq [n]$ of size at least $n/ (r+1)$ subspaces have the same dimension. 

By \cref{lemma:sunflower}, the family of subspaces $S=\{V_i : i\in J\}$ contains a sunflower with at least
\[
s=\left(\frac{n}{(r+1)\cdot \exp(\frac{r^2+r-2}{2})}\right)^{1/(r+1)}
\]
petals. Let $I\seq J$ denote the sub-collection that form a sunflower, where $\card{I}\ge s$.

Let $C = \bigcap_{i\in I} V_i$ denote the core of the sunflower. Note that $C$ is itself a (linear) subspace. We call a linear map $\alpha \colon  C\to \F_2$ an \emph{assignment}. Let $C^*$ denote the set of all assignments, where $\card{C^*}=2^{\dim(C)}$.

Note that after fixing the core $C$ by a given assignment, the distribution $\cP_I\rest\alpha=(P_i\rest\alpha : i \in I)$ is now pairwise independent, and by \cref{cor:chebyshev-lemma-bounded-rank1}
\begin{equation}\label{eq:after-rest-alpha}
\norm{\cP_I\rest\alpha - \cB_{1/3}^{\ox s}}_\tv \ge 1 - O(2^{r-\dim(C)}/s).
\end{equation}
Since $\E_{\alpha\sim C^*}[\cP_I\rest\alpha]=\cP_I$, we get
\begin{align*}
\norm{\cP-\cB_{1/3}^{\ox n}}_\tv &\ge \norm{\cP_I - \cB_{1/3}^{\ox s}}_\tv \\
&= \norm{\E_{\alpha\sim C^*}[\cP_I\rest\alpha] - \cB_{1/3}^{\ox s}}_\tv \\
&\ge 1 - 2\sum_{\alpha \in C^*}(1-\norm{\cP_I\rest\alpha - \cB_{1/3}^{\ox s}}_\tv) \tag{By \cref{lem:dist-of-convex-comb}}\\
&\ge 1 - 2^{\dim(C)+1} \cdot O(2^{r-\dim(C)}/s)\tag{By \cref{eq:after-rest-alpha}}\\
&\ge 1 - O(2^r/s)\\
&= 1 - O\Big( \frac{(r+1)^{1/(r+1)}\exp\big(\frac{r^2+2r-2}{2r+2}\big)}{n^{1/(r+1)}} \Big)\\
&\ge 1 - O\Big(\frac{2^{r/2}}{n^{1/(r+1)}}\Big) \tag{$x^{1/x}\le e^{1/e}$ for $x \ge 1$},\\
\end{align*}
as required. (Note that the hidden constants in big-$O$ are absolute.)
\end{proof}

\section{Lower bound against quadratic distributions}\label{sec:degree-2-lower-bound}

\begin{theorem}\label{thm:quadratic-sources}
    Let $\cQ=(Q_1,\dots,Q_n)$ be a distribution, where each $Q_i$ is a quadratic polynomial. Then
    \[
    \norm{\cQ-\cB_{1/3}^{\ox n}}_\tv \ge 1- n^{-\Omega(1/\log\log(n))}.
    \]
\end{theorem}

\begin{proof}
    Let $r > 0$ be a parameter to be determined later.
    
    If at least $n/2$ many quadratics have $\rank_1(Q_i)\le 4r$, the apply \cref{cor:quadratics-of-bounded-rank-1} to get a distance of $1-\exp(-c2^{-Cr\log r}n)$ for some constants $C, c > 0$. $(\star)$
    
    Otherwise, at least $n/2$ many $Q_i$'s, which we assume are $Q_1,\dots, Q_{n/2}$, satisfy $\rank(Q_i) > 4r$. Construct a graph $G=(V, E)$ with vertex set $V=[n/2]$. We add an edge $ij\in E$ if the corresponding polynomials $Q_i,Q_j$ have small-rank difference, that is, whenever $\rank_1(Q_i - Q_j) \leq 4r$.
    
    Let $w$ be another parameter whose value will be chosen later.
    
    We consider two cases for this graph.
    \paragraph{Case 1 (large independent set). } The graph $G$ contains an independent set $I\seq V$ of size $w$. Then by \cref{cor:chebyshev-only-bounded-rank-2} we have that
    \[
    \norm{\cQ-\cB_{1/3}^{\ox n}}_\tv \ge \norm{\cQ_I-\cB_{1/3}^{\ox \card{I}}}_\tv \ge 1 - O(\frac{1}{w}) - O(\frac{1}{2^r}),
    \]
    where big-O is hiding absolute constants.

    \paragraph{Case 2 (large degree). } There is some vertex $i^*\in V$ whose degree in $G$ is at least $\lfloor\frac{n}{2w}\rfloor$. Let $S=N_G(i^*)$ be the set of neighbors of $i^*$. Let $s$ be equal to $\card{S}$, which is $\Omega(n/w)$.
    
    Since the difference of $Q_{i^*}$ and any of its neighbors has small rank, we may write
    \[
    Q_j= Q_{i^*} + \Gamma_j(L_{j1},\dots, L_{j, 4r}),
    \]
     for all $j\in S$ and some $4r$ linear functions.

    We are now in the setting of \cref{lem:dist-tv-common-Q}.
    Indeed, note that for any fixed $z \in \Bits$, the distribution
    \[
    \cR_z \defeq (\Gamma_j(L_{j1},\dots, L_{j,4r}) + z \; :\; j\in S) 
    \]
    has $\rank_1$ at most $4r$ in each coordinate. Hence, by our bounded $\rank_1$ theorem (\cref{thm:bounded-rank-1}, \cref{cor:quadratics-of-bounded-rank-1}) we have the following for some constants $c, C>0$.
    \begin{equation}\label{eq:R-z}
    \norm{\cR_z - \cB_{1/3}^{\ox s}}_\tv \ge 1- \exp\Big(-c2^{-C r\log r} n/w\Big).
    \end{equation}

    If $Q_{i^*}$ is constant then \cref{eq:R-z} is already a lower bound for $\cQ$. Else, $Q_{i^*}$ has two possible outcomes,
    and $\Pr[Q_{i^*} = z] \geq \tau = 1/4$ for each $z \in \supp(Q_{i^*})$, and by \cref{lem:dist-tv-common-Q} we have
    \[
    \norm{(Q_i)_{i \in S} - \cB_{1/3}^{\ox s}}_\tv
    \geq 1-  \frac{8}{\tau} \cdot \exp\Big(-c2^{-C r\log r} n/w\Big) 
    \geq 1 -  32 \cdot \exp\Big(-c2^{-C r\log r} n/w\Big).
    \]

    \paragraph{Combining cases 1 \& 2. } For some absolute constant $c>0$ we have,
    \[
    \norm{\cQ - \cB_{1/3}^{\ox n}}_\tv
    \ge 1 - \max\Big\{
    \underbracket{O(\frac{1}{w}) + O(\frac{1}{2^r})}_{\text{Case 1}};\;\; 
    \underbracket{32\exp\Big(-c2^{-C r\log r} n/w\Big)}_{\text{Case 2}}\Big\}.
    \]
    Since the error from Case 2 subsumes the error at $(\star)$, we omit the latter. It remains to choose $r=r(n)$ and $w=w(n)$ so that the maximum error is minimized.

    A straightforward calculation shows that for some constant $c'>0$, setting $r=c'\log n/\log\log n$ and $w=2^r$, the Case 1 error term eventually becomes larger than Case 2. It then follows that
    \[
    \norm{\cQ - \cB_{1/3}^{\ox n}}_\tv \ge 1 - O(\frac{1}{2^r}) \ge 1 - \exp(-c'\log n/\log\log n + O(1)) \ge 1- n^{-\Omega(1/\log\log n)},
    \]
    as required.
\end{proof}

\section{Lower bounds against distributions of bounded quadratic rank}\label{sec:bound-rank2-lower-bound}

\begin{theorem}\label{thm:bounded-rank-2}
    Let $r \geq 2$, and let $\cP = (P_1,\dots,P_n)$ be an n-tuple cubic polynomials $P_i \colon \F_2^m \to \F_2$, such that $\rank_2(P_i) \leq r \leq \sqrt{\log\log(n)}$ for all $i \in [n]$.
    Suppose that  $\abs{\Pr_X[P_i(X) = 1] - 1/3} \geq \delta$ for all $i \in [n]$.
    Then $\norm{\cP - \cB_{1/3}^{\ox n}}_\tv\ge 1 - \frac{\exp(-\log(n)^{c/r^2})}{\delta^2}$ 
    for some absolute constant $c>0$.
\end{theorem}

\begin{remark}
By \cref{lemma:deg-3-dyadic} we have an explicit lower bound $\delta \geq \delta_{2,r} = \frac{1}{12} \cdot 2^{-((r+2) \cdot 2^r + r)}>0$.
In the special case when all $P_i$'s are polynomials of degree 3, we take $\delta = \delta_3 >0$ to be an absolute constant.
\end{remark}

We will use the following notation throughout this section. For a function $P\colon \F_2^m \to \F_2$, we write \defemph{$\rank_{2,1}(P) \le (r,s)$} if there exist quadratic polynomials $Q_1,\dots,Q_r$, linear polynomials $L_1,\dots,L_s$, and a function $\Gamma$ such that $P = \Gamma(Q_1,\dots,Q_r,L_1,\dots,L_s)$.

Before proving the theorem, we need the following procedure, an explicit case of \emph{regularization} for degree $2$ factors, which is summarized in the following claim.

\begin{claim}\label{claim:regularize-deg-2}
    Fix $C > 1$.
    Let $P \colon \F_2^m \to \F_2$ be a function of $\rank_{2,1}(P) \leq (r,s)$.
    Then, there exists some $0 \leq k \leq r$ such that $P$ can be written as
    \begin{equation*}
        P(x) = \Gamma'(\cQ(x);\cL(x)),
    \end{equation*}
    where
    \begin{itemize}
        \item either
            \begin{inparaenum}[(i)]
            \item $\cQ = \emptyset$, and
            \item $\cL$ is a set of linear functions of size $\abs{\cL} \leq (C+1)^r \cdot (r+s) - r$.
            \end{inparaenum}
        \item or
            \begin{inparaenum}[(i)]
            \item $\cQ$ is a subset of $\{Q_1(x),\dots,Q_r(x)\}$ of size $r-k$,
            \item $\cL$ is a set of linear functions of size $\abs{\cL} \leq (C+1)^k \cdot (r+s) - r$, and
            \item $\rank_1(\cQ) > C (C+1)^k \cdot (r+s)$.
            \end{inparaenum}
    \end{itemize}

\end{claim}

\begin{proof}
Write $P \colon \F_2^m \to \F_2$ as
    \begin{equation*}
        P(x) = \Gamma(Q_1(x),\dots,Q_r(x);L_1(x),\dots,L_s(x)),
    \end{equation*}
where $Q_i$'s are quadratic polynomials and $L_i$'s are linear functions.    

Define the collections $\cQ$ and $\cL$ as in \cref{alg:regularize-each-Qi}.
\begin{algorithm}[h]
    \caption{Regularizing $\cQ$}
    \label{alg:regularize-each-Qi}
    
    \begin{algorithmic}[1]
        \State Initialize $\cL = \{L_1,\dots,L_s\}$.
        \State Initialize $I = \{1,\dots,r\}$.
            \Comment{{\footnotesize The set $I \seq \{1,\dots,r\}$ will correspond to $\cQ = \{Q_i : i \in I\}$.}}
        
        \State Set $k = 0$.
            \Comment{{\footnotesize In each iteration we will have $\abs{I} = r-k$ and $\abs{\cL} \leq (C+1)^k \cdot (r+s)$.}}
        
        \While{$I \neq \emptyset$ and $\rank_1(\cQ) \leq C (C+1)^k \cdot (r+s)$}
        
            \State Let $(a_i)_{i \in I}$ be such that $\rank_1(\sum_{i \in I} a_i Q_i) < C (C+1)^k \cdot (r+s)$.
    
            \State Augment $\cL$ by adding to it at most $C (C+1)^k \cdot (r+s)$ linear functions so that $\sum_{i \in I} a_i Q_i$ can be expressed using the functions in $\cL$. If $\cL$ has linear dependencies, remove them.

            \State Remove from $I$ one of the $i$'s for which $a_i = 1$.
            
            \State Redefine $\Gamma$ to be a function that takes $\cQ$ and $\cL$ and computes $P_i$.

            \State Increment $k$ by 1.
        \EndWhile
    \end{algorithmic}
\end{algorithm}

    Next we show that the obtained sets $\cL$ and $\cQ = \{Q_i : i \in I\}$ satisfy the properties stated in the claim.

    Consider the value of $k$ at the end of the procedure, and note that $k$ is equal to the number of iterations performed by the while loop.

    \medskip
    Note first that we decrease the size of $I$ by one in each iteration, and hence in the end $\cQ$ is a subset of $\{Q_1(x),\dots,Q_r(x)\}$ of size $\abs{I} = r-k$.

    We show an upper bound on the size of $\cL$ by induction on $k$. For $k=0$ we have $\card{\cL} \leq s = (C+1)^k(r+s)-r$ by the assumption of the claim. Suppose the bound $\abs{\cL} \leq (C+1)^{k-1} \cdot (r+s) - r$ holds just before the last iteration. Then, in the last iteration iteration we increase the size of $\cL$ by at most $C (C+1)^{k-1} \cdot (r+s)$, and hence the size of $\cL$ is now upper bounded by
    \begin{equation*}
        \big((C+1)^{k-1} \cdot (r+s) - r\big) + \big(C (C+1)^{k-1} \cdot (r+s)\big)
        = (C+1)^{k} \cdot (r+s) - r.
    \end{equation*}

    Finally, by the stopping condition of the while loop, if $k = r$, then $\cQ = \emptyset$, and otherwise $\rank_1(\cQ) \geq C (C+1)^k \cdot (r+s)$, as required.
\end{proof}

We are now ready to prove \cref{thm:bounded-rank-2}.

\begin{proof}[Proof of \cref{thm:bounded-rank-2}.]
Let $\eps(n,r,s) > 0$ be such that for any $n$-tuple $\cP = (P_1,\dots,P_n)$ of functions $P_i \colon \F_2^m \to \F_2$, with $\rank_{2,1}(P_i) \leq (r,s)$ satisfies $\norm{\cP-\cB_{1/3}^{\ox n}}_\tv \geq 1 - \eps(n,r,s)$.
Below we derive a bound on $\eps$ using induction.

For the base case, note that by \cref{thm:bounded-rank-1} we have $\eps(n,r=0,s) \leq \exp\Big(-\frac{c^s\delta^{2s}}{s!}\cdot n + s\Big)$ for some absolute constant $c>0$.

Next we bound $\eps(n,r,s)$ for $r \geq 1$ and arbitrary $s \geq 0$.
Let $\cP = (P_1,\dots,P_n)$ be an n-tuple of functions $P_i \colon \F_2^m \to \F_2$, such that $\rank_{2,1}(P_i) \leq (r,s)$ for all $i \in [n]$.

We start by regularizing each $P_i$ as in \cref{claim:regularize-deg-2} with some parameter $C \geq 8$ to be chosen later.
After the regularization, for each $P_i$ there is some $0 \leq k_i \leq r$ such that
$P_i = \Gamma_i(\cQ^{(i)},\cL^{(i)})$, where $\cQ^{(i)}$ is a set of quadratic functions of size at most $r-k_i$, $\cL^{(i)}$ is a set of linear functions, $\abs{\cQ^{(i)}} + \abs{\cL^{(i)}} \leq (C+1)^{k_i} \cdot (r+s)$, and $\rank_1(\cQ^{(i)}) > C \cdot (C+1)^{k_i} \cdot (r+s)$.

\medskip

Let $k^*$ be the most popular among $k_i$'s, and let $S_0 \seq [n]$ be a subset of size $\geq n/(r+1)$ such that $\abs{\cQ^{(i)}} + \abs{\cL^{(i)}} \leq (C+1)^{k^*} \cdot (r+s)$ and $\rank_1(\cQ^{(i)}) > C \cdot (C+1)^{k^*} \cdot (r+s)$ for all $i \in S_0$.

We consider the following three cases.

\begin{itemize}
    \item Suppose that $k^* \geq 1$. Then $\rank_{2,1}(P_i) \leq (r-k^*, (C+1)^{k^*} \cdot (r+s))$ for all $i \in S_0$. And by the induction hypothesis we get
    \begin{equation}\label{eq:rank-2-bound1}
        \norm{\cP-\cB_{1/3}^{\ox n}}_\tv \geq 1 - \eps(\frac{n}{r+1}, r-k^*, (C+1)^{k^*} \cdot (r+s)).
    \end{equation}

    \item Assume now that $k^*=0$.
    Let $w \in \N$ be a parameter to be chosen later, and suppose there is a subset $S \seq S_0$ of size $\abs{S} = w$ such that for all $i \neq i'$ in $S$ it holds that $\rank_1(Q^{(i)} \cup Q^{(i')}) > C(r+s)$. Then, by applying \cref{cor:chebyshev-pairwise-high-rank-lin-part} we get
    \begin{equation}\label{eq:rank-2-bound2}
        \norm{\cP-\cB_{1/3}^{\ox n}}_\tv
        \geq 1 - \frac{4^s}{\delta^2 \cdot w^{1/(s+1)}} - \frac{2^{-(C-3)(r+s)/2}}{\delta^2}
        \geq 1 - \frac{4^s}{\delta^2 \cdot w^{1/(s+1)}} - \frac{2^{-C(r+s)/4}}{\delta^2}.
    \end{equation}
    
    \item Otherwise, we have $k^* = 0$ and there exists $i^* \in [n]$ and a subset $S \seq [n]$ of size $\abs{S} \geq (1/2^r) \cdot (\abs{S_0}/w) $ such that 
    $\rank_1(Q^{(i^*)} \cup Q^{(i)}) \leq C(r+s)$, and furthermore
    the bounded rank is realized with the \emph{same linear combination in $Q^{(i^*)}$}. That is, there exists some linear combination $\sum_{j=1}^r \beta_j Q^{(i^*)}_j$ in  such that for each $i \in S$ we have
    $\rank_1(\sum_{j=1}^r \alpha^{(i)} Q^{(i)}_j + \sum_{j=1}^r \beta_j Q^{(i^*)}_j)) < C(r+s)$. 

    Hence, for each $i \in S$, we can remove one of the $Q^{(i)}_j$ from $\cQ^{(i)}$, augment $\cL^{(i)}$ by adding to it at most $C(r+s)$ linear functions, and rewrite $\Gamma_i$ as a function that depends on the new $\cQ^{(i)}$ (of size at most $r-1$), new $\cL^{(i)}$ and $Q^{*} = \sum_{j=1}^r \beta_j Q^{(i^*)}_j$. We emphasize that $Q^{*}$ is the same for all $i \in S$.

    That is, for each $i \in S$, we can write $P_i$ as
    \begin{equation*}
        P_i = \Gamma_i(Q^{(i)}_1,\dots,Q^{(i)}_{r-1};\cL^{(i)};Q^{*}).
    \end{equation*}
    For $z \in \Bits$ denote by $P_i^{z} = \Gamma_i(Q^{(i)}_1,\dots,Q^{(i)}_{r-1};\cL^{(i)};z).$
    Note that $\rank_{2,1}(P_i^{z}) \leq (r-1, s + C(r+s))$.
    Hence, by the induction hypothesis, we have
    \begin{equation*}
        \norm{(P_i^{z})_{i \in S}-\cB_{1/3}^{\ox \abs{S}}}_\tv \geq 1 - \eps(\abs{S}, r-1, s + C(r+s)).
    \end{equation*}
    Therefore, since $\deg(Q^*) \leq 2$, each $z \in \supp(Q^*) \seq \Bits$ is obtained with probability at least $1/4$, and thus, by \cref{lem:dist-tv-common-Q} we get
    \begin{equation}\label{eq:rank-2-bound3}
        \norm{(P_i)_{i \in S}-\cB_{1/3}^{\ox \abs{S}}}_\tv \geq 1 - 32\eps(\frac{n}{w \cdot (r+1) \cdot 2^r}, r-1, s + C(r+s)).
    \end{equation}
\end{itemize}

Combining \cref{eq:rank-2-bound1,eq:rank-2-bound2,eq:rank-2-bound3}, we get
\begin{equation*}
    \eps(n,r,s) \leq \max \Big(
    \max_{k^*}\eps(\frac{n}{r+1},r-k^*,(C+1)^{k^*} (r+s));
    \ \
    \frac{4^s}{\delta^2 \cdot w^{1/(s+1)}} + \frac{2^{-C(r+s)/4}}{\delta^2};
    \ \
    32\eps(\frac{n}{w \cdot (r+1) \cdot 2^r}, r-1, s+ C(r+s))
    \Big).
\end{equation*}

Using monotonicity, we can bound $\eps(n,r,s)$ by
\begin{equation}\label{eq:rank2-eps-induction}
    \eps(n,r,s)
    \leq \max \Big(
    \frac{4^s}{\delta^2 \cdot w^{1/(s+1)}} + \frac{2^{-C(r+s)/4}}{\delta^2};
    \ \
    32\eps(\frac{n}{w \cdot (r+1) \cdot 2^r}, r-1, (C+1)^r \cdot (r+s))
    \Big)
\end{equation}
with the base case
\begin{equation}\label{eq:rank2-eps-base-case}
\eps(n,r=0,s) \leq \exp\Big(-\frac{c^s\delta^{2s}}{s!}\cdot n + s\Big).
\end{equation}

We have the following upper bound on $\eps$.
\begin{claim}\label{claim:rank-2-eps-solution}
    Let $\eps(n,r,s)$ be defined inductively by \cref{eq:rank2-eps-induction,eq:rank2-eps-base-case}.
    Then,
    \begin{equation*}
        \eps(n,r,s)
        \leq    
        32^r \cdot \max \Big(
        \term;
        \quad
        \exp \left( - \frac{(c\delta)^{h^*}}{h^*!} \cdot \frac{n}{w^r \cdot (r+1)! \cdot 2^{r(r+1)/2}} +  h^* \right)
        \Big)
    \end{equation*}
    for $h^* \leq 2(C+1)^{r^2} \cdot (2r+s)$ and $\term = \frac{4^s}{\delta^2 \cdot w^{1/(s+1)}} + \frac{2^{-C(r+s)/4}}{\delta^2}$.

\end{claim}
Next, we set the parameters $w = n^{1/2r}$ and $C = \log(n)^{1/3r^2}-1$.
Hence, if $r,s < \sqrt{\log(n)\log(n)}$, then
$h^* \leq (\log(n)^{1/3r^2})^{r^2} \cdot (2r+s) = \log(n)^{\frac13} \cdot \sqrt{\log(n)\log(n)}
< \sqrt{\log(n)}$.
Therefore, for $r \geq 2$ we have
\begin{eqnarray*}
    \eps(n,r,s)
    & \leq &
    32^r \cdot \max \Big(
    \term;
    \quad
    \exp \left( - \frac{(c\delta)^{h^*}}{h^*!} \cdot \frac{n}{w^r \cdot (r+1)! \cdot 2^{r(r+1)/2}} +  h^* \right)
    \Big) \\
    & \leq &
    32^r \cdot \max \Big(
    \frac{4^{\sqrt{\log(n)\log(n)}}}{\delta^2 \cdot n^{\Omega(1/rs)}} + \frac{\exp(-\log(n)^{\Omega(1/r^2)})}{\delta^2}) ;
    \quad
    \exp \left( - \frac{(c\delta)^{\sqrt{\log(n)}}}{\sqrt{\log(n)}!} \cdot \frac{\sqrt{n}}{2^{r^2}} +  \log(n)^{1/6} \right)
    \Big) \\
\end{eqnarray*}
for some absolute constant $c>0$.

Recall, by \cref{lemma:deg-3-dyadic} we have $\delta \geq \frac{1}{12} \cdot 2^{-((r+2) \cdot 2^r + r)} \geq 2^{-C^{\sqrt{\log\log(n)}}}$ for some constant $C$ (say, $C=6$ works). Hence the dominant term in the displayed equation about is $\frac{\exp(-\log(n)^{\Omega(1/r^2)})}{\delta^2}$.
Therefore
    \begin{equation}
        \norm{(\cP-\cB_{1/3}^{\ox n}}_\tv
        \geq  \norm{(P_i)_{i \in S}-\cB_{1/3}^{\ox \abs{S}}}_\tv
        \geq 1 - \eps(n,r,s)
        \geq \exp(-\log(n)^{c/r^2})/\delta^2
    \end{equation}
for some absolute constant $c>0$, as required.
\end{proof}

We now return to the proof of \cref{claim:rank-2-eps-solution}.
\begin{proof}[Proof of \cref{claim:rank-2-eps-solution}]
Define $h_{r,s} = \sum_{i=1}^r (C+1)^{\sum_{j=1}^i j}\cdot i + (C+1)^{\sum_{j=1}^r j} \cdot s$.
Note that
\begin{itemize}
    \item $h_{0,s} = s$,
    \item $h_{r,s} = h_{r-1,(C+1)^r(r+s)}$,
    \item $h_{r,s} \leq (C+1)^{r^2} \cdot (2r+s)$.
\end{itemize}

We prove by induction that
\begin{equation}\label{eq:eps-closed-bound}
    \eps(n,r,s) \leq 32^r \cdot \max 
    \Big(
        \term;
        \quad
        \exp \left( - \frac{c^{h_{r,s}} \cdot \delta^{2h_{r,s}}}{h_{r,s}!} \cdot \frac{n}{w^r \cdot (r+1)! \cdot 2^{r(r+1)/2}} +  h_{r,s} \right)
    \Big)
\end{equation}
The claim follow by setting $h^* = 2h_{r,s}$.

For the base case of $r=0$ we have $h_{0,s} = s$, and hence $\eps(n,r=0,s) \leq \exp(-\frac{c^s \cdot \delta^{2s}}{s!}\cdot n + s)  = \exp(-\frac{c^{h_{0,s}}\cdot \delta^{2h_{0,s}}}{h_{0,s}!} \cdot n + h_{0,s})$.

For the induction step, we assume that the bound on $\eps$ holds for $r-1$, and  prove it for $r$. Indeed, by \cref{eq:rank2-eps-induction} we have
\begin{eqnarray*}
    \eps(n,r,s) 
    & \leq & \max \Big(
    \term;
    \ \
    32 \cdot \eps(\frac{n}{w \cdot (r+1) \cdot 2^r}, r-1, (C+1)^r \cdot (r+s))
    \Big) \\
    & = & \max \Big(
    \term;
    \ \
    32 \cdot \eps(n', r-1, s')
    \Big),
\end{eqnarray*}
where $n' = \frac{n}{w \cdot (r+1) \cdot 2^r}$ and $s' = (C+1)^r \cdot (r+s)$.
By the induction hypothesis we have
\begin{equation*}
    \eps(n',r-1,s') \leq
    32^{r-1} \cdot \max \Big(
    \term;
    \quad
     \exp \left( - \frac{c^{h_{r-1,s'}} \cdot\delta^{2 h_{r-1,s'}}}{h_{r-1,s'}!} \cdot \frac{n'}{w^{r-1} \cdot r! \cdot 2^{(r-1)r/2}} +  h_{r-1,s'} \right)
    \Big)
\end{equation*}
Plugging it into the bound on $\eps(n,r,s)$ we get
\begin{eqnarray*}
    \eps(n,r,s)
    & \leq &
    32^r \cdot \max \Big(
    \term;
    \quad
     \exp \left( - \frac{c^{h_{r-1,s'}}\cdot \delta^{2h_{r-1,s'}}}{h_{r-1,s'}!} \cdot \frac{n'}{w^{r-1} \cdot r! \cdot 2^{(r-1)r/2}} +  h_{r-1,s'} \right)
    \Big) \\
    & = &
    32^r \cdot \max \Big(
    \term;
    \quad
     \exp \left( - \frac{c^{h_{r,s}} \cdot \delta^{2 h_{r,s}}}{h_{r,s}!} \cdot \frac{n}{w^r \cdot (r+1)! \cdot 2^{r(r+1)/2}} +  h_{r,s} \right)
    \Big),
\end{eqnarray*}
where the last equality uses the fact that $h_{r,s} = h_{r-1,s'}$.
This completes the proof of \cref{eq:eps-closed-bound}.
\end{proof}

\section{Lower bounds against distributions of degree three}\label{sec:degree-3-lower-bound}

\begin{theorem}\label{thm:cubics}
    Let $\cP = (P_1,\dots,P_n)$ be an n-tuple of cubic polynomials $P_i \colon \F_2^m \to \F_2$.
    Then 
    \[
    \norm{\cP - \cB_{1/3}^{\ox n}}_\tv\ge 1 - \exp(-c\sqrt{\log\log(n)}).
    \]
\end{theorem}

\begin{proof}
    We have several cases.

    \begin{itemize}
        \item There is a subset $S_0 \seq[n]$ of size at least $n/2$ such that $\rank_2(P_i) \leq r$ for all $i \in S_0$. This reduces to \cref{thm:bounded-rank-2}, and we get $\norm{\cP - \Ber(1/3)^{\ox n}}_\tv\ge 1 - \exp(-\log(n)^{c/r^2})$.

        \item Otherwise, there is a subset $S_0 \seq[n]$ of size at least $n/2$ such that $\rank_2(P_i) > r$ for all $i \in S_0$. Is there is $S \seq S_0$ of size at least $w$ such that $\rank_2(P_i + P_{i'}) > r$ for all $i \neq i' \in S$. The we apply Chebyshev, and get 
        $\norm{\cP - \Ber(1/3)^{\ox n}}_\tv\ge 1 - \frac{c}{w} - \frac{c}{\exp(r)}$.

        \item Otherwise, there exists some $i^* \in S_0$ and $S \seq S_0$ of size at least $\abs{S_0}/w$ such that $\rank_2(P_i + P_{i^*}) \leq r$ for all $i \in S$.
        This means for all $i \in S$ we can write $P_i = \Gamma_i(\cQ^{(i)}) + P_{i^*}$ for some $\Gamma_i \colon \Bits^r \to \Bits$ and $\cQ^{(i)}$ consisting of $r$ quadratic polynomials. This gives the bound
        $\norm{\cP - \Ber(1/3)^{\ox n}}_\tv\ge 1 - \exp(-\log(n/w)^{c/r^2})$.
    \end{itemize}
    Taking $w = \sqrt{n}$ and $r = c \sqrt{\log\log(n)}$ for some $0<c<1/4$, we get that
    \begin{eqnarray*}
        \norm{\cP - \Ber(1/3)^{\ox n}}_\tv
        & \geq & 1 - \max\left\{
        \exp(
        \log(n)^{c/r^2});
        \quad 
        \frac{c}{w} + \frac{c}{\exp(r)};
        \quad
        \exp(-\log(n/w)^{c/r^2})
        \right\} \\
        & = & 1 - \frac{c}{\exp(r)} \\
        & = & 1 - \exp(-c\sqrt{\log\log(n)})
    \end{eqnarray*}
    for some absolute constant $c>0$.
\end{proof}

\section{Lower bound against degree $d$ distributions}\label{sec:deg-d-lower-bound}

In this section we show that any polynomial distribution $\cP$ of degree $d$ is $1-o_n(1)$ far from $\cB_{1/3}^{\ox n}$, where $o_n(1)$ is a vanishing function of $n$ for any constant $d$.

We begin with the following proposition, and then proceed to the proof of the main theorem of this section.

\begin{proposition}\label{prop:density-of-high-rank-variety}
    Let $\cF = (P_1,\dots, P_K)$ be a degree $d$ factor whose rank is $c_\KL(d, 2^{-K}\eta)$ for some $\eta\in (0, 1]$, and let $Z$ denote the set of common zeros of the $P_i$. Then the density of $Z$, i.e., $\Pr_X[X\in Z]$, is equal to $2^{-K} \pm \eta$.
\end{proposition}
\begin{proof}
    Observe that $\Pr[P_1(X)=\dots = P_K(X) = 0] = \E[(-1)^{Y_1P_1(X) + \dots + Y_KP_K(X)}]$ where $X$ is drawn uniformly at random and the $Y_i\sim \cU$ are i.i.d. variables mutually independent from $X$. Then we have the following.
    \begin{align*}
        \E[(-1)^{Y_1P_1 + \dots + Y_KP_K}] &= 2^{-K}\sum_{S\seq [K]} \E\Big[(-1)^{\sum_{i\in S} P_i}\Big] \\
        &= 2^{-K} + \sum_{S\seq [K], S\neq \emptyset} (\pm2^{-K}\eta)\\
        &= 2^{-K} \pm \eta. \qedhere
    \end{align*}
\end{proof}

\begin{theorem}[Degree-$d$]\label{thm:deg-d-lower-bound}
For every $0 < \eps \le 1$ and $d\in \N$ there exists some $\eps',d' > 0$ such that the following holds. Let $\cP = (P_1,\dots, P_n)$ be a distribution generated by polynomials of degree $d$. Then, 
\[
\norm{\cP - \cB_{1/3}^{\ox n}}_\tv \ge 1 - O(d'n^{-\eps'}) - O(\eps),
\]
where the hidden constants in big $O$ are absolute. It follows that,
\[
\norm{\cP - \cB_{1/3}^{\ox n}}_\tv \ge 1 - o_n(1),
\]
where $o_n(1)$ is a vanishing function of $n$ for any constant $d$.
\end{theorem}
\begin{proof}
    Define the collection of factors $\cF_1,\dots, \cF_n$ with $\cF_i = \{P_i\}$. Let $f\colon \N\to \N$ be given by $f(r) = c_\KL(d, 2^{-3r}\eps\delta^2)$, where $\delta=\delta_d$ is from \cref{thm:exp-low-deg-far-from-third}.
    
    We invoke \cref{thm:lossy-pairwise-regularization} to obtain a collection of factors $(\cG_i : i\in I)$ where $I\seq [n]$ has size at least $\card{I} = s \ge n^{\eps'}$, and $\cG_i$ are pairwise $f$-regular. For simplicity, assume that $I=[s]$. Note that \cref{thm:lossy-pairwise-regularization} guarantees $\card{\cG_i}\le r$ where $r$ only depends on $d,\eps$.

    Let $\cC = \bigcap_{i\in I}\cG_i = \{C_1,\dots, C_\ell\}$ be the common intersection, and write each $\cG_i$ as the disjoint union $\cC \sqcup \cQ_i$. Now we can write the joint distribution of the set $I$ as,
    \[
    \cP_I=(P_1,\dots, P_s)=(\Gamma_1(\cC, \cQ_1), \dots, \Gamma_s(\cC, \cQ_s)). 
    \]
    
    Let $\Gamma(\cQ_1,\dots, \cQ_s, \cC)$ denote $(\Gamma_1(\cC, \cQ_1), \dots, \Gamma_s(\cC, \cQ_s))$. Consider $\sigma\colon \cC\to\F_2$ that assigns values to each polynomial in $\cC$. Let us call $\sigma$ an assignment, noting that there are at most $2^\ell$ such assignments. Because of $f$-regularity and by \cref{prop:density-of-high-rank-variety}, $\Pr[\cC = \sigma(\cC)] \ge 2^{-\ell} - 2^\ell2^{-3r}\eps \ge 2^{-r-1}$.

    By \cref{lem:chebyshev-bounded-rank-d} (with $\eta = 2^{-2r}\eps\delta^2$), we have that for any fixed assignment $\sigma$,
    \begin{equation}\label{eq:a-single-sigma-assignment}
    \norm{\Gamma_{\sigma}(\cQ_1,\dots, \cQ_s) - \cB_{1/3}^{\ox s}}_\tv \ge 1 - O(\frac{1}{\delta^2 s}) - O(2^{-2r}\eps).
    \end{equation}
    Now we apply \cref{lem:dist-tv-common-Q} (with $\card{S}= 2^{\ell}$, and $\tau = 2^{-r-1}$) to $\Gamma(\cQ_1,\dots, \cQ_s, \cC)$, where the common part is captured by $\cC$.
    \begin{align*}
        \norm{\cP - \cB_{1/3}^{\ox n}}_\tv &\ge \norm{\cP_I - \cB_{1/3}^{\ox s}}_\tv\\
        &\ge 1 - 4\cdot 2^\ell \cdot 2^{r+1} \cdot \Big(O(\frac{1}{\delta^2 s}) + O(2^{-2r}\eps)\Big)\tag{\cref{lem:dist-tv-common-Q} and \cref{eq:a-single-sigma-assignment}}\\
        &\ge 1 - O(\frac{2^{2r}}{\delta^2 n^{\eps'}}) - O(\eps).\tag{$2^\ell \le 2^r$}
    \end{align*}
    For constants $d$ and $\eps$, the term $O(\frac{2^{2r}}{\delta^2 n^{\eps'}})$ vanishes as $n$ grows. Since $\eps$ can be made arbitrarily small, it follows that
    \[
    \norm{\cP - \cB_{1/3}^{\ox n}}_\tv \ge 1 - o_{n}(1).\qedhere
    \]
\end{proof}

\section{Acknowledgments}

We thank Mark Kahn for providing the examples of low degree polynomials with $\Pr[P(x)=1]$ surprisingly close to 1/3 described in \cref{sec:open-problems}.

\newpage
\bibliographystyle{alpha}
\bibliography{refs}

@article{hatami2019higher,
  title={Higher-order fourier analysis and applications},
  author={Hatami, Hamed and Hatami, Pooya and Lovett, Shachar},
  journal={Foundations and Trends{\textregistered} in Theoretical Computer Science},
  volume={13},
  number={4},
  pages={247--448},
  year={2019}
}

@book{LidlNiederreiter1996,
  title={Finite Fields},
  author={Lidl, Rudolf and Niederreiter, Harald},
  year={1996},
  place={Cambridge},
  edition={2},
  series={Encyclopedia of Mathematics and its Applications},
  publisher={Cambridge University Press},
  collection={Encyclopedia of Mathematics and its Applications}
}

@inproceedings{haramaty2010structure,
  title={On the structure of cubic and quartic polynomials},
  author={Haramaty, Elad and Shpilka, Amir},
  booktitle={Proceedings of the forty-second ACM symposium on Theory of computing},
  pages={331--340},
  year={2010}
}

@inproceedings{bhattacharyya2014algorithmic,
  title={Algorithmic regularity for polynomials and applications},
  author={Bhattacharyya, Arnab and Hatami, Pooya and Tulsiani, Madhur},
  booktitle={Proceedings of the Twenty-Sixth Annual ACM-SIAM Symposium on Discrete Algorithms},
  pages={1870--1889},
  year={2014},
  organization={SIAM}
}

@article{horacsek2025sourcedecoding,
	title = {Constant-time source decoding},
	author = {Horacsek, Jordan and Lee, Chin Ho and Shinkar, Igor and Viola, Emanuele and Zhou, Renfei},
    journal = {Electronic Colloquium on Computational Complexity (ECCC)},
    year    = {2025},
    note    = {TR25-164},
}

@inproceedings{kaufman2008worst,
  title={Worst case to average case reductions for polynomials},
  author={Kaufman, Tali and Lovett, Shachar},
  booktitle={2008 49th Annual IEEE Symposium on Foundations of Computer Science},
  pages={166--175},
  year={2008},
  organization={IEEE}
}

@article{green2007distribution,
  title={The distribution of polynomials over finite fields, with applications to the Gowers norms},
  author={Green, Ben and Tao, Terence},
  journal = {Contributions to Discrete Mathematics},
  volume  = {4},
  number  = {2},
  pages   = {1--36},
  year    = {2009},
  doi     = {10.11575/CDM.V4I2.62086}
}

@inproceedings{viola2023new,
  title={New sampling lower bounds via the separator},
  author={Viola, Emanuele},
  booktitle={38th Computational Complexity Conference (CCC 2023)},
  pages={26--1},
  year={2023},
  organization={Schloss Dagstuhl--Leibniz-Zentrum f{\"u}r Informatik}
}

@article{Viob,
  title={The complexity of distributions},
  author={Viola, Emanuele},
  journal={SIAM Journal on Computing},
  volume={41},
  number={1},
  pages={191--218},
  year={2012},
  publisher={SIAM}
}

@article{kane2025symmetric,
  title={Symmetric Distributions from Shallow Circuits},
  author={Kane, Daniel M and Ostuni, Anthony and Wu, Kewen},
  journal={arXiv preprint arXiv:2511.14127},
  year={2025}
}

@inproceedings{shaltiel2024explicit,
  title={Explicit codes for poly-size circuits and functions that are hard to sample on low entropy distributions},
  author={Shaltiel, Ronen and Silbak, Jad},
  booktitle={Proceedings of the 56th Annual ACM Symposium on Theory of Computing},
  pages={2028--2038},
  year={2024}
}

@inproceedings{kane2025locally,
  title={Locally sampleable uniform symmetric distributions},
  author={Kane, Daniel M and Ostuni, Anthony and Wu, Kewen},
  booktitle={Proceedings of the 57th Annual ACM Symposium on Theory of Computing},
  pages={1807--1816},
  year={2025}
}

@inproceedings{watts2023unconditional,
  title={Unconditional quantum advantage for sampling with shallow circuits},
  author={Watts, Adam Bene and Parham, Natalie},
  booktitle =	{17th Innovations in Theoretical Computer Science Conference (ITCS 2026)},
  pages =	{17:1--17:12},
  series =	{Leibniz International Proceedings in Informatics (LIPIcs)},
  ISBN =	{978-3-95977-410-9},
  ISSN =	{1868-8969},
  year =	{2026},
  volume =	{362},
  editor =	{Saraf, Shubhangi},
  publisher =	{Schloss Dagstuhl -- Leibniz-Zentrum f{\"u}r Informatik},
  address =	{Dagstuhl, Germany},
  URL =		{https://drops.dagstuhl.de/entities/document/10.4230/LIPIcs.ITCS.2026.17},
  URN =		{urn:nbn:de:0030-drops-253048},
  doi =		{10.4230/LIPIcs.ITCS.2026.17}
}

@inproceedings{yu2024sampling,
  title={Sampling, flowers and communication},
  author={Yu, Huacheng and Zhan, Wei},
  booktitle={15th Innovations in Theoretical Computer Science Conference (ITCS 2024)},
  pages={100--1},
  year={2024},
  organization={Schloss Dagstuhl--Leibniz-Zentrum f{\"u}r Informatik}
}

@article{erdos1960intersection,
  title={Intersection theorems for systems of sets},
  author={Erd{\"o}s, Paul and Rado, Richard},
  journal={Journal of the London Mathematical Society},
  volume={1},
  number={1},
  pages={85--90},
  year={1960},
  publisher={Wiley Online Library}
}

@article{kasami1976weight,
  title={On the weight enumeration of weights less than 2.5 d of Reed—Muller codes},
  author={Kasami, Tadao and Tokura, Nobuki and Azumi, Saburo},
  journal={Information and control},
  volume={30},
  number={4},
  pages={380--395},
  year={1976},
  publisher={Elsevier}
}

@article{kasami1970weight,
  title={On the weight structure of Reed-Muller codes},
  author={Kasami, Tadao and Tokura, Nobuki},
  journal={IEEE Transactions on Information Theory},
  volume={16},
  number={6},
  pages={752--759},
  year={1970},
  publisher={IEEE}
}

@article{kaufman2012weight,
  title={Weight distribution and list-decoding size of reed--muller codes},
  author={Kaufman, Tali and Lovett, Shachar and Porat, Ely},
  journal={IEEE transactions on information theory},
  volume={58},
  number={5},
  pages={2689--2696},
  year={2012},
  publisher={IEEE}
}

@article{furst1984parity,
  title={Parity, circuits, and the polynomial-time hierarchy},
  author={Furst, Merrick and Saxe, James B and Sipser, Michael},
  journal={Mathematical systems theory},
  volume={17},
  number={1},
  pages={13--27},
  year={1984},
  publisher={Springer}
}

@article{ajtai198311,
  title={{$\Sigma_1^1$-formulae on finite structures}},
  author={Ajtai, Mikl{\'o}s},
  journal={Annals of pure and applied logic},
  volume={24},
  number={1},
  pages={1--48},
  year={1983},
  publisher={Elsevier}
}

@inproceedings{hastad1986almost,
  title={Almost optimal lower bounds for small depth circuits},
  author={Hastad, John},
  booktitle={Proceedings of the eighteenth annual ACM symposium on Theory of computing},
  pages={6--20},
  year={1986}
}

@inproceedings{kane2024locality,
  title={Locality bounds for sampling hamming slices},
  author={Kane, Daniel M and Ostuni, Anthony and Wu, Kewen},
  booktitle={Proceedings of the 56th Annual ACM Symposium on Theory of Computing},
  pages={1279--1286},
  year={2024}
}

@inproceedings{filmus2023sampling,
    author =	{Filmus, Yuval and Leigh, Itai and Riazanov, Artur and Sokolov, Dmitry},
    title =	{{Sampling and Certifying Symmetric Functions}},
    booktitle =	{Approximation, Randomization, and Combinatorial Optimization. Algorithms and Techniques (APPROX/RANDOM 2023)},
    pages =	{36:1--36:21},
    series =	{Leibniz International Proceedings in Informatics (LIPIcs)},
    ISBN =	{978-3-95977-296-9},
    ISSN =	{1868-8969},
    year =	{2023},
    volume =	{275},
    editor =	{Megow, Nicole and Smith, Adam},
    publisher =	{Schloss Dagstuhl -- Leibniz-Zentrum f{\"u}r Informatik},
    address =	{Dagstuhl, Germany},
    URL =		{https://drops.dagstuhl.de/entities/document/10.4230/LIPIcs.APPROX/RANDOM.2023.36},
    URN =		{urn:nbn:de:0030-drops-188611},
    doi =		{10.4230/LIPIcs.APPROX/RANDOM.2023.36}
}

@inproceedings{chattopadhyay2022space,
  title={The space complexity of sampling},
  author={Chattopadhyay, Eshan and Goodman, Jesse and Zuckerman, David},
  booktitle={13th Innovations in Theoretical Computer Science Conference (ITCS 2022)},
  pages={40--1},
  year={2022},
  organization={Schloss Dagstuhl--Leibniz-Zentrum f{\"u}r Informatik}
}

@article{goos2020lower,
  title={A lower bound for sampling disjoint sets},
  author={G{\"o}{\"o}s, Mika and Watson, Thomas},
  journal={ACM Transactions on Computation Theory (TOCT)},
  volume={12},
  number={3},
  pages={1--13},
  year={2020},
  publisher={ACM New York, NY, USA}
}

@article{viola2020sampling,
  title={Sampling lower bounds: boolean average-case and permutations},
  author={Viola, Emanuele},
  journal={SIAM Journal on Computing},
  volume={49},
  number={1},
  pages={119--137},
  year={2020},
  publisher={SIAM}
}

@article{viola2016quadratic,
  title={Quadratic maps are hard to sample},
  author={Viola, Emanuele},
  journal={ACM Transactions on Computation Theory (TOCT)},
  volume={8},
  number={4},
  pages={1--4},
  year={2016},
  publisher={ACM New York, NY, USA}
}

@article{de2012extractors,
  title={Extractors and lower bounds for locally samplable sources},
  author={De, Anindya and Watson, Thomas},
  journal={ACM Transactions on Computation Theory (TOCT)},
  volume={4},
  number={1},
  pages={1--21},
  year={2012},
  publisher={ACM New York, NY, USA}
}

@inproceedings{beck2012large,
  title={Large deviation bounds for decision trees and sampling lower bounds for AC0-circuits},
  author={Beck, Chris and Impagliazzo, Russell and Lovett, Shachar},
  booktitle={2012 IEEE 53rd Annual Symposium on Foundations of Computer Science},
  pages={101--110},
  year={2012},
  organization={IEEE}
}

@inproceedings{lovett2011bounded,
  title={Bounded-depth circuits cannot sample good codes},
  author={Lovett, Shachar and Viola, Emanuele},
  booktitle={2011 IEEE 26th Annual Conference on Computational Complexity},
  pages={243--251},
  year={2011},
  organization={IEEE}
}

\newpage
\appendix
\section{One polynomial's distance from Ber(non-dyadic)}\label{sec:non-constructive-lower-bound-for-a-single-poly}

In this section we prove \cref{thm:non-dyadic-gap} which states that for any non-dyadic number $\rho\in(0,1)$ and degree $d$ polynomial $P$
\[
\abs{\Pr_X[P(X)=1] - \rho} \ge \Omega_{d,\rho}(1).
\]
We prove this by showing that if one has a sequence of polynomials $P^{(n)}$ converging to a distribution $D$, then the limit distribution $D$ must have probability weights that are dyadic numbers. This means that there is no sequence of polynomials that samples $\Ber(\rho)$ better than a constant distance as the probability weight $\rho$ is not a dyadic number. Alas, the proof is non-constructive and does not elucidate what the constant $\Omega_{d,\rho}(1)$ should be. The proof relies on a key regularization lemma that uses similar techniques to \cite{green2007distribution, kaufman2008worst} but instead is applied to a \emph{sequence} of factors.  Then the bias-rank theorem of Kaufman and Lovett (\cref{thm:kaufman-lovett}) and Vazirani's XOR lemma (\cref{lem:vazirani-xor}) are used to conclude the result.

We will be using the standard concepts defined in \cref{sec:factors}, and we will extend those definitions to sequences of factors.

\paragraph{Sequences.}
A sequence of objects will be written $\{A^{(n)}\}_{n\in\N}$, and we often suppress the braces and index set, writing simply $A^{(n)}$.
Given two sequences $A^{(n)}$ and $B^{(n)}$, we say that $B^{(n)}$ is a \emph{subsequence} of $A^{(n)}$ if there exists a strictly increasing map $f\colon \N\to\N$ such that
\[
B^{(n)} = A^{(f(n))}\qquad \text{for all } n\in\N .
\]
In this case we write $B^{(n)} \prec A^{(n)}$.

\paragraph{Sequences of factors.}
We will frequently consider sequences of factors $\cF^{(n)}$.
Unless stated otherwise, it is understood that all factors in a given sequence have the same degree bound $d$ and the same dimension vector $(M_1,\dots,M_d)$.
Accordingly, we decorate by a superscript $(n)$ all data that may vary with $n$ (while $d$ and $(M_1,\dots,M_d)$ remain fixed); for instance,
\[
\cF^{(n)} = (P^{(n)}_{1},\dots,P^{(n)}_{K}).
\]

\begin{definition}[Refinement of a Sequence of Factors]
We say a sequence of factors $\cG^{(n)}$ is a refinement of $\cF^{(n)}$ (or that $\cG^{(n)}$ \emph{refines} $\cF^{(n)}$) if there is a sub-sequence $\cF'^{(n)}\prec \cF^{(n)}$ and a $\Gamma\colon \F_2^L \to \F_2^K$ such that $\Gamma(\cG^{(n)}) = \cF'^{(n)}$ for all $n$, where $L = \dim \cG^{(n)}$ and $K = \dim \cF^{(n)}$ (observe that $\Gamma, L$, and $K$ are the same for all $n$). Refining is transitive: if $\cH^{(n)}$ refines $\cG^{(n)}$ and $\cG^{(n)}$ refines $\cF^{(n)}$, then $\cH^{(n)}$ refines $\cF^{(n)}$.
\end{definition}

\begin{definition}[Regularity of a Sequence of Factors]
    Let $\cF^{(n)}$ denote a sequence of factors. We say this sequence is $r$-regular if there exists some $N_0\in \N$ such that for all $n\ge N_0$ the factor $\cF^{(n)}$ is $r$-regular. If $\cF^{(n)}$ is $r$-regular for every $r\in \N$ then we say this sequence is \emph{fully regular}.
\end{definition}

\begin{lemma}[Key Lemma]\label{lem:fully-regular-sub-sequence}
    Let $\cF^{(n)}$ be a sequence of degree $d$ factors with dimension vector $(M_1,\dots, M_d)$. Then there exists a fully regular sequence of degree $d$ factors $\cH^{(n)}$ that is a refinement of $\cF^{(n)}$.
\end{lemma}
\begin{proof}
    The proof goes by strong induction on the dimension vector $(M_1,\dots, M_d)\in \N^d$. The set $\N^d$ is a well-ordered set by the inverse lexicographic order\footnote{We have $(M'_1, \dots, M'_d) < (M_1,\dots, M_d)$ if there exists some $i\in[d]$ such that $M'_j = M_j$ for all $j > i$ and $M'_i < M_i$.}. As the base case of induction we have that any factor with dimension vector $(M_1,\dots, M_d) = (1,0, \dots, 0)$ is evidently fully regular.
    
    Let $\cF^{(n)}=(P^{(n)}_1,\dots, P^{(n)}_K)$ where $K = M_1 + \dots + M_d$. If this sequence is fully regular then we are done. Otherwise, there exists some $r\in \N$ such that for infinitely many $n$ the factor $\cF^{(n)}$ is not $r$-regular. Let $\cF'^{(n)}\prec \cF^{(n)}$ denote this subsequence so that every factor in $\cF'^{(n)}$ is \emph{not} $r$-regular.

    Now, for simplicity of notation assume $n$ as fixed, and consider some specific factor $\cF' = (P'_1,\dots, P'_K)$ in the sequence $\cF'^{(n)}$. Since $\cF'$ is not $r$-regular, there exist $\lambda_1,\dots, \lambda_K\in \F_2$ not all zeros, $\ell \in [d], j\in [K], \Gamma \colon  \F_2^r \to \F_2$ and polynomials $Q_1,\dots, Q_r$ of degree at most $\ell-1$ such that
    \[
    \lambda_1P'_1 + \dots +\lambda_KP'_K = \Gamma(Q_1, \dots, Q_r),
    \]
    where $\ell = \max_i \deg(\lambda_iP'_i)$ and $j$ is such that $\deg(\lambda_jP'_j) = \ell$. 

    Let us refer to the tuple $(\lambda_1,\dots,\lambda_K,\ell,j,\Gamma,\deg(Q_1),\dots,\deg(Q_r))$ as a \emph{certificate of non-regularity}.
    Since the number of possible certificates is bounded by a function of $r$, $d$, and $K$, it follows that infinitely many factors in the sequence $\cF'^{(n)}$ share the same certificate. Let $\cF''^{(n)}\prec \cF'^{(n)}$ be a sub-sequence for which the certificate of non-regularity is the same.
    
    Now by removing the $j$\textsuperscript{th} polynomial in $\cF''^{(n)}$ and adding $Q^{(n)}_i$'s we then refine 
    \[
    \cF''^{(n)} = (P''^{(n)}_1,\dots, P''^{(n)}_K) \quad\longmapsto\quad \cG^{(n)} = (P''^{(n)}_1, \dots, P''^{(n)}_{j-1}, P''^{(n)}_{j+1}, \dots, P''^{(n)}_K, Q^{(n)}_1,\dots, Q^{(n)}_{r}). 
    \]
    (If $\ell=1$, we simply remove $P_j$ and add no $Q_i$'s because $Q_i$'s would be constant in that case). Since the certificate was chosen in such a way that it is the same for all factors in $\cF''^{(n)}$, there exists a function $\Psi\colon \F_2^{K+r-1}\to\F_2^{K}$ such that $\Psi(\cG^{(n)}) = \cF''^{(n)}$ for all $n$. Note that the dimension vector of $\cG^{(n)}$ is at most $(M_1, \dots, M_{\ell-1} + r, M_\ell-1,  M_{\ell+1}, \dots, M_d)$ which is smaller than $(M_1,\dots, M_d)$ in our ordering of $\N^d$. By the induction hypothesis there exists some fully regular $\cH^{(n)}$ that refines $\cG^{(n)}$. Finally, by transitivity of refinement, $\cH^{(n)}$ also refines $\cF^{(n)}$.
\end{proof}

\begin{lemma}\label{lem:fully-regular-means-uniform}
    Let $\cF^{(n)} = (P^{(n)}_1,\dots, P^{(n)} _K)$ be a fully regular sequence of factors of degree $d$. Then the joint distribution of polynomials $(P^{(n)}_1,\dots, P^{(n)} _K)$ converges to the uniform distribution $U_K$.
\end{lemma}
\begin{proof}
    We show that for every positive $\eps$ there exists some $N_0$ such that for all $n\ge N_0$ we have $\norm{\cF^{(n)} - U_K}_\tv \le \eps$, which implies our claim.
    
    Let $\eps > 0$ be fixed. Then define $r = c_\KL(d, \eps2^{-K/2})$, where $c_\KL$ comes from \cref{thm:kaufman-lovett}. Since $\cF^{(n)}$ is fully regular there exists some $N_0$ such that $\cF^{(n)}$ is $r$-regular for all $n \ge N_0$. Then by \cref{thm:kaufman-lovett}
    \[
    \bias(\lambda_1P^{(n)}_1 + \dots + \lambda_KP^{(n)}_K) \le \eps2^{-K/2},\qquad \text{for all } n\ge N_0 \text{ and } \lambda_i \text{ not all zeros.}
    \]
    Thus, by Vazirani's XOR lemma, we have that $\norm{\cF^{(n)}- U_K}_\tv \le \eps$ for all $n\ge N_0$.
\end{proof}

\begin{definition}[Computable distribution]
    We say a distribution $D$ over $\Bits^s$ is \defemph{computable} if there exists some $K\in\N$ and function $\Gamma\colon  \Bits^K\to\Bits^s$ such that $D = \Gamma(U_K)$.
\end{definition}

It is an easy observation that a distribution $D$ is computable if and only if the probability of every outcome is a dyadic number.

\begin{lemma}\label{lem:computability-of-limit}
    Let $P^{(n)}$ be a sequence of polynomials of degree $d$ that converges to the distribution $D$ in $\tv$-distance. Then, there exists a positive integer $K$ and a function $\Gamma\colon \F_2^K\to\F_2$ such that $D=\Gamma(U_K)$. In other words, the limit distribution $D$ is computable.
\end{lemma}
\begin{proof}
    Let us describe $P^{(n)}$ by the trivial sequence of factors that computes it; that is to say $\cF^{(n)} = (P^{(n)})$ with the dimension vector $(0, \dots, 0, 1)\in \N^d$.

    By \cref{lem:fully-regular-sub-sequence} there exists $K$ and a fully regular degree $d$ sequence $\cH^{(n)}=(Q_1^{(n)}, \dots, Q_K^{(n)})$ of dimension $K$ that is a refinement of $\cF^{(n)}$. That is, for a subsequence $P'^{(n)}$ of $P^{(n)}$ there exists some $\Gamma$ such that for all $n$, $P'^{(n)} = \Gamma(Q_1^{(n)}, \dots, Q_K^{(n)})$.

    Since $\cH^{(n)}$ is fully regular, by \cref{lem:fully-regular-means-uniform} the joint distribution $(Q_1^{(n)}, \dots, Q_K^{(n)})$ converges to the uniform distribution, and therefore, $P'^{(n)} = \Gamma(Q^{(n)}_1,\dots,Q_K^{(n)})$ converges to $\Gamma(U_K)$.
\end{proof}

\begin{remark}
    In fact, \cref{lem:computability-of-limit} extends with essentially no additional effort to the setting of joint distributions of low-degree polynomials. We have presented it in the present form purely for the sake of slightly cleaner notation.
\end{remark}

\begin{theorem}\label{thm:non-dyadic-gap}
    Let $0 < \rho < 1$ be a \emph{non-dyadic} number, and let $d$ be a positive integer. Then for any degree $d$ polynomial $P\colon \F_2^n\to\F_2$ it holds that $\norm{P - \Ber(\rho)}_\tv \ge \Omega_{d, \rho}(1)$.
\end{theorem}
\begin{proof}
    Toward contradiction, suppose this were not the case. That is, suppose there exists a sequence $P^{(n)}$ of polynomials of degree $d$ such that
    \[
    \norm{P^{(n)} - \Ber(\rho)}_\tv \le o_{d, \rho}(1) \qquad \text{ for } n\to\infty.
    \]
    Then this sequence converges to $\Ber(\rho)$ in total variation distance. On the other hand, by \cref{lem:computability-of-limit} there exist $K$ and $\Gamma$ such that the limit of $P^{(n)}$ is the distribution $\Gamma(U_K)$. This is a contradiction as $\rho$ is non-dyadic.
\end{proof}

\section{Deferred proofs}\label{sec:deferred-proofs}

\begin{proof}[Proof of Vazirani's XOR lemma (\cref{lem:vazirani-xor})]
Let $\mu,u\colon \F_{2}^n\to [0,1]$ be the probability mass functions of $\cX$ and $\cU^{\ox n}$ respectively. We will bound $\norm{\cX - \cU^{\ox n}}_2$ from above. We have
\begin{align*}
\norm{\cX - \cU^{\ox n}}^2_2 &= \sum_{x}(\mu(x)-u(x))^2 \\
&= 2^n\E_{x}\big[(\mu(x)-u(x))^2\big] \\
&= 2^{n}\E\big[\mu(x)^2\big]+ 2^n\E\big[u(x)^2\big]-2^{n+1}\E\big[\mu(x)u(x)\big] \\
&= 2^n\sum_{\gamma}\hat{\mu}(\gamma)^2 + 2^{-n} - 2\hat{\mu}(0) \\
&= 2^{n} \sum_{\gamma\neq 0} \hat{\mu}(\gamma)^2\tag{because $\hat{\mu}(0)=2^{-n}$}.
\end{align*}

Consider some non-trivial character $\chi_\gamma = (-1)^L$ where $L$ is some nonzero linear form. Then, $\abs{\hat{\mu}(\gamma)}=\abs{\E_x[\chi_\gamma(x) \mu(x)]}=\abs{2^{-n} \bias(L(X))}\le \eps\cdot 2^{-n}$. Therefore,
\begin{align*}
\norm{\cX - \cU^{\ox n}}^2_2 &\le 2^{2n-1}(2^{-2n}\eps^2) \le \eps^2.
\end{align*}
Then $\norm{\cX - \cU^{\ox n}}_{1}\leq 2^{n/2}\eps$ is immediate by Cauchy-Schwarz.
\end{proof}

\begin{proof}[Proof of \cref{prop:properties-of-rapid-function}]
   We prove this claim by induction on $\vec M\defeq (M_1,\dots, M_d)$, where the tuples in $\N^d$ are ordered by inverse lexicographic ordering. Our base case covers all cases where $N_2=\dots = N_d = 0$. We have,
   \[
   \psi(N_1,0\dots, 0) = N_1 \le \sum_i M_i \le  \psi(M_1,\dots, M_d).
   \]
   From now on we assume that $N_i > 0$ for some $i\ge 2$. If $\vec N = \vec M$, then the statement clearly holds. Thus we further assume that $\vec N < \vec M$.
   
   Let $j$ be the right-most index such that $M_j > N_j$ (i.e. $j$ is the index that shows $\vec M$ is greater than $\vec N$ in inverse lexicographic order). In particular, the entry $M_j$ is nonzero.
   
   Define $\ell\ge 2$ (resp. $k\ge 2$) to be the first index of $\vec M$ (resp. $\vec N$) that is nonzero. Note that $\ell \le j$.

   Now by definition of $\psi$, we have $\psi(\vec M) = \psi(\vec M')$ and $\psi(\vec N)= \psi(\vec N')$ where,
   \begin{align*}
       \vec M' &= (M_1,\dots, M_{\ell-1}+f(\textstyle\sum_iM_i), M_\ell - 1, \dots, M_d)\\
       \vec N' &= (N_1,\dots, N_{k-1}+f(\textstyle\sum_iN_i), N_k - 1, \dots, N_d)
   \end{align*}
   
    Clearly, $\sum_i\vec N_i' \le \sum_i\vec M_i'$ still holds because of monotonicity of $f$. To show that $\vec N' < \vec M'$, we consider three cases.
   \begin{itemize}
       \item ($k \ge j$) If $k \ge j$ then $N'_k < M'_k$ and therefore $\vec N' < \vec M'$.
       \item ($k < j$ and $\ell <j$) Then $N'_j=N_j < M_j= M'_j$ and $\vec N' < \vec M'$.
       \item ($k < j$ and $\ell = j$) Then $M'_j\ge N'_j$, and if this inequality is strict, then we are done. Otherwise, $M'_{j-1} \ge f(\sum_i M_i) \ge \sum_i M_i > N'_{j-1}$.
   \end{itemize}
   The three cases show that $\vec M' > \vec N'$ and the claim follows by the induction hypothesis.
\end{proof}

\end{document}